\documentclass[lettersize,journal]{IEEEtran}
\usepackage{amsfonts}
\usepackage{amsmath}
\usepackage{amssymb}
\usepackage{epsfig}
\usepackage{multicol}
\usepackage{booktabs}
\usepackage{lettrine}
\usepackage{float}
\usepackage{footmisc}
\usepackage{epstopdf}
\usepackage[compress]{cite}
\usepackage{times}
\usepackage{stfloats}
\usepackage{mathtools}
\usepackage{algorithm}
\usepackage{algorithmic}
\usepackage{multirow}
\usepackage{color}
\usepackage{subfigure}
\usepackage{diagbox}
\usepackage[caption=false,font=normalsize,labelfont=sf,textfont=sf]{subfig}
\usepackage{makecell}
\usepackage{colortbl}
\usepackage{xcolor}
\usepackage{array}
\usepackage{arydshln}
\newtheorem{assumption}{Assumption}

\newtheorem{theorem}{Theorem}
\newtheorem{remark}{Remark}
\newtheorem{lemma}{Lemma}
\newtheorem{proposition}{Proposition}
\newtheorem{definition}{Definition}
\newtheorem{corollary}{Corollary}

\allowdisplaybreaks[4]
\input{tcilatex}
\begin{document}
\title{Critic-Free Policy Iteration for Continuous-Time Zero-Sum Games: A Policy-Space Riccati Approach
\thanks{
}}
\author{Jiacheng Wu, {Yang~Zhu}\IEEEauthorrefmark{1}, ~Hongye~Su,
\IEEEmembership{Senior Member,~IEEE}

\thanks{
Jiacheng Wu, Yang Zhu, and Hongye Su are with State Key Laboratory of Industrial Control Technology,
Institute of Cyber-Systems and Control, Zhejiang University, Hangzhou
310027, China (e-mail: jiachengwu@zju.edu.cn;
zhuyang88@zju.edu.cn;
hysu@iipc.zju.edu.cn;).}

\thanks{\textsuperscript{*}Corresponding author: Yang Zhu.}%

}
\maketitle
\begin{abstract}
This paper develops a critic-free policy iteration (PI) method for continuous-time linear zero-sum games. The central idea is to characterize the saddle-point policies directly in the joint policy space, rather than treating the quadratic value matrix as an iterative variable. A policy game Riccati equation (PGRE) is introduced whose unknowns are policy gains only. Its solutions are shown to be in one-to-one correspondence with the symmetric solutions of the game algebraic Riccati equation. Based on a stabilizing anchor, PI is performed directly in the actor space. The actor-space Jacobian is nonsingular at every stabilizing policy, and the resulting policy sequence coincides with that of simultaneous PI. For unknown dynamics, a data-driven algorithm uses a single batch of data and nullspace projection of endpoint increments to eliminate the value matrix, yielding an actor-only regression. A necessary and sufficient rank condition for unique policy recovery is established and shown to be iteration-invariant, demonstrating that critic identifiability is unnecessary. A power systems frequency-regulation example verifies convergence and policy recovery, while scalability tests demonstrate substantial reductions in computational and memory requirements.
\end{abstract}
\section{Introduction}
\IEEEPARstart{R}OBUST control of dynamical systems subject to adversarial or energy-bounded disturbances can be formulated as a two-player zero-sum differential game, in which a controller minimizes a prescribed performance index while a disturbance acts to maximize it \cite{basar2008hinfinity,kiumarsi2017h,luo2021policyc}. For continuous-time linear systems with quadratic performance \cite{li2025reinforcementc,jiang2023reinforcement}, the saddle-point control and disturbance policies are characterized by a stabilizing solution of the game algebraic Riccati equation (GARE). This characterization establishes a fundamental connection among $H_{\infty }$ control, differential games, and optimal state-feedback design. When an accurate system model is available, the saddle-point solution can be computed using model-based Riccati equation methods \cite{kleinman1968iterative,nortmann2026iterative}. In many practical applications \cite{LianBS2024BOOK,lin2025datadriven,liu2026modelfree,yang2024hamiltoniannnls}, the system matrices are unavailable, motivating reinforcement learning (RL) and adaptive dynamic programming approaches \cite{wu2026accelerated,liu2020adaptive,tan2025prescribed,Jha2019IEEETAC} that learn the saddle-point policies directly from measured system trajectories.

Policy iteration (PI) has provided a systematic basis for data-driven optimal control of continuous-time systems. Early adaptive-critic PI schemes learned linear-quadratic feedback policies with partial model knowledge \cite{Vrabie2009}, while subsequent data-driven formulations reused fixed trajectory data to address systems with completely unknown dynamics \cite{Jiang2012}. For zero-sum games, integral RL was developed to recover saddle-point policies from measured trajectories \cite{VrabieLewis2011,Li2014,shen2025data}. Wu \emph{et al.} \cite{WuLuo2013} proposed a simultaneous policy update algorithm (SPUA), in which the control and disturbance policies are improved within a single iteration loop and convergence is established through a Newton sequence for the GARE. Off-policy extensions further allowed the behavior data to be generated independently of the policies being evaluated \cite{LuoWuHuang2015,lopez2023efficient,zhao2024novel,yang2023modelfree}. Despite their different model requirements and data-acquisition protocols, these methods typically retain the value function, or its quadratic matrix parameter, during policy evaluation. Their data-driven implementations therefore lead to actor–critic regressions in which the policy gains and critic parameters are estimated jointly.

From an implementation viewpoint, the critic is not part of the final 
feedback law. Nevertheless, for an $n$-dimensional system, a quadratic 
value function introduces $d_n=n(n+1)/2$ critic parameters in addition 
to the $sn=(m+q)n$ joint-policy parameters. Conventional off-policy PI 
therefore estimates $d_n+sn$ unknowns and imposes a full-rank condition 
on the joint actor--critic regression 
\cite{Jiang2012,Li2014,wu2025distributed}. Policy-only algebraic formulations have recently 
been derived for LQR \cite{sassano2025policyac} and $H_\infty$ control \cite{Sassano2024}. For LQR, 
the policy Riccati equation admits a Newton iteration identical to 
Kleinman's method \cite{PossieriSassano2026}, and endpoint projection can 
produce a data-driven actor-only update \cite{WuZhuSu2026}. However, the zero-sum 
setting couples the minimizing and maximizing policies through 
an indefinite weighting matrix. It further requires the policy equation 
to preserve the GARE solution structure and the data condition to 
distinguish policy recovery from critic identification. These issues are 
not resolved by the single-player LQR constructions, while the existing 
$H_\infty$ policy equation does not provide an iterative or data-driven 
critic-free scheme.

Motivated by these observations, this article develops a critic-free PI 
framework for continuous-time linear zero-sum games. The control and 
disturbance gains are combined into a joint policy variable, and a 
quadratic policy game Riccati equation (PGRE) is constructed through a 
Sylvester lifting. In contrast to the Hamiltonian-observability 
construction in \cite{Sassano2024}, the PGRE establishes a direct 
correspondence between its policy solutions and the symmetric solutions 
of the GARE. A stabilizing anchor then removes the open-loop spectral 
condition and enables Newton updates directly in the actor space, 
extending the policy-space iteration developed for LQR in 
\cite{PossieriSassano2026}. Each update is shown to be well posed at a 
stabilizing policy and to coincide with the model-based SPUA in 
\cite{WuLuo2013}. For unknown dynamics, the endpoint-projection mechanism 
introduced for LQR in \cite{WuZhuSu2026} is extended to the joint 
min--max policy. Applying this projection to an integral system identity 
eliminates the value-matrix term and yields an actor-only regression from 
a fixed batch of data. The resulting representation also separates the 
rank condition for policy recovery from that for critic identification.

The main contributions of this article are summarized as follows. 
\begin{itemize}
\item[1)] We develop a policy-space formulation for continuous-time linear zero-sum games by introducing a PGRE in the joint control–disturbance gains. The PGRE preserves the full symmetric solution structure of the GARE without requiring a matching condition between the two channels, while its anchored form enables a well-posed Newton iteration directly in the actor space despite the indefinite game weighting.

\item[2)] We develop a data-driven critic-free PI scheme for continuous-time linear zero-sum games with unknown dynamics. By projecting the integral policy-evaluation identity onto the nullspace of the endpoint increments, the value-matrix term is eliminated before regression, yielding an actor-only update involving only the joint policy gains. The resulting update is lossless with respect to the actor variables and reproduces the model-based PI.

\item[3)] We establish a policy-oriented data-informativity theory that characterizes when the saddle-point policy can be uniquely recovered from data. A necessary and sufficient actor rank condition is derived and separated from critic identifiability, and is further shown to be invariant across all stabilizing iterations. This reveals that unique policy recovery does not require the value matrix to be identifiable.
\end{itemize}

\emph{Notation:}
Throughout this article, let $\mathbb{R}$ and $\mathbb{R}^{a\times b}$ denote 
the sets of real numbers and real $a\times b$ matrices, respectively, 
and let $\mathbb{S}^{a}$ denote the set of real symmetric 
$a\times a$ matrices.
For a matrix $\mathbf{X}$, $\mathbf{X}^{\top}$, 
$\operatorname{tr}(\mathbf{X})$, $\operatorname{spec}(\mathbf{X})$, 
$\operatorname{rank}(\mathbf{X})$, $\ker(\mathbf{X})$, and 
$\operatorname{im}(\mathbf{X})$ denote its transpose, trace, spectrum, 
rank, kernel, and image, respectively. The orthogonal complement of a 
subspace $\mathcal{V}$ is denoted by $\mathcal{V}^{\perp}$. 
The operators $\operatorname{col}(\cdot)$ and 
$\operatorname{diag}(\cdot)$ denote vertical concatenation and block 
diagonalization, respectively. For a matrix $\mathbf{X}$, 
$\operatorname{vec}(\mathbf{X})$ is obtained by stacking its columns, 
and $\operatorname{unvec}_{a\times b}(\cdot)$ denotes the inverse 
vectorization onto $\mathbb{R}^{a\times b}$. For 
$\mathbf{X}\in\mathbb{S}^{n}$, 
$\operatorname{vech}(\mathbf{X})\in\mathbb{R}^{d_n}$ stacks its 
lower-triangular entries, where $d_n=n(n+1)/2$. The trace-dual 
vectorization $\operatorname{svec}(\cdot)$ is defined by
$
  \operatorname{svec}(\mathbf{P})^{\top}
  \operatorname{vech}(\mathbf{X})
  =
  \operatorname{tr}(\mathbf{P}\mathbf{X}),
$
$
  \mathbf{P},\mathbf{X}\in\mathbb{S}^{n}.
$
The symbols $\otimes$ and $\oplus$ denote the Kronecker product and 
Kronecker sum, respectively. 
 For a Fr\'echet-differentiable matrix-valued mapping 
$\mathcal{F}$, $D\mathcal{F}(\mathbf{X})[\mathbf{E}]$ denotes its 
derivative at $\mathbf{X}$ in the direction $\mathbf{E}$.

\section{System Model and Zero-Sum Game}
\label{sec:system_game}

Consider the continuous-time linear system
\begin{equation}
    \dot{x}(t)
=\mathbf{A}x(t)+\mathbf{B}_{1}u(t)+\mathbf{B}_{2}w(t),
    \qquad x(0)=x_{0},
\label{1}
\end{equation}
where $x(t)\in\mathbb{R}^{n}$ is the state,
$u(t)\in\mathbb{R}^{m}$ is the control input, and
$w(t)\in\mathbb{R}^{q}$ is the exogenous disturbance, with
$w(\cdot)\in L_{2}[0,\infty)$. The system matrices
$\mathbf{A}\in\mathbb{R}^{n\times n}$,
$\mathbf{B}_{1}\in\mathbb{R}^{n\times m}$, and
$\mathbf{B}_{2}\in\mathbb{R}^{n\times q}$ are unknown.
Let $\mathbf{Q}=\mathbf{Q}^{\top}\geq0$ and
$\mathbf{R}=\mathbf{R}^{\top}>0$. For a prescribed attenuation level
$\gamma>0$, define the running payoff
$
    r(x(t),u(t),w(t))
    \triangleq x^{\top}(t)\mathbf{Q}x(t)
    +u^{\top}(t)\mathbf{R}u(t)-\gamma^{2}w^{\top}(t)w(t)
$
and the infinite-horizon payoff
\begin{equation}
    J(x_{0},u,w)
    \triangleq\int_{0}^{\infty}
    r\bigl(x(t),u(t),w(t)\bigr)\,dt.
\label{2}
\end{equation}
For an admissible pair $(u,w)$, the associated cost-to-go is
\begin{equation}
    V\bigl(x(t)\bigr)
    \triangleq\int_{t}^{\infty}
    r\bigl(x(\tau),u(\tau),w(\tau)\bigr)\,d\tau.
\label{3}
\end{equation}
Let $\mathcal{U}$ and $\mathcal{W}$ denote the admissible classes of
causal state-feedback control and disturbance policies, respectively. The upper value
of the associated two-player zero-sum game is
\begin{equation*}
    V^{\ast}(x_{0})
    =\inf_{u\in\mathcal{U}}\sup_{w\in\mathcal{W}}
      J(x_{0},u,w).
\end{equation*}
A policy pair $(u^{\ast},w^{\ast})$ is a saddle point if
\begin{equation*}
    J(x_{0},u^{\ast},w)
    \leq J(x_{0},u^{\ast},w^{\ast})
    \leq J(x_{0},u,w^{\ast})
\end{equation*}
for every $u\in\mathcal{U}$ and $w\in\mathcal{W}$.
Assume that $(\mathbf{A},\mathbf{B}_{1})$ is stabilizable and
$(\mathbf{Q}^{1/2},\mathbf{A})$ is detectable. Let
$\gamma^{\ast}$ denote the infimum of the attenuation levels for
which a stabilizing saddle point exists \cite{lewis2012optimal}. Fix a feasible level
$\gamma>\gamma^{\ast}$, and let
$\mathbf{P}^{\ast}=\mathbf{P}^{\ast\top}\geq0$ denote a stabilizing
solution of the game algebraic Riccati equation (GARE)
\begin{equation}
\begin{aligned}
    \mathbf{0}={}&
    \mathbf{A}^{\top}\mathbf{P}^{\ast}
    +\mathbf{P}^{\ast}\mathbf{A}
    +\mathbf{Q}
    -\mathbf{P}^{\ast}\mathbf{B}_{1}\mathbf{R}^{-1}
     \mathbf{B}_{1}^{\top}\mathbf{P}^{\ast}\\
&+\gamma^{-2}\mathbf{P}^{\ast}\mathbf{B}_{2}
     \mathbf{B}_{2}^{\top}\mathbf{P}^{\ast}.
\end{aligned}
\label{4}
\end{equation}
The corresponding saddle-point policies are
\begin{align}
    u^{\ast}(t)
    &=-\mathbf{K}^{\ast}x(t)
      =-\mathbf{R}^{-1}\mathbf{B}_{1}^{\top}
       \mathbf{P}^{\ast}x(t),
    \label{5}\\
    w^{\ast}(t)
    &=\mathbf{L}^{\ast}x(t)
      =\gamma^{-2}\mathbf{B}_{2}^{\top}
       \mathbf{P}^{\ast}x(t).
    \label{6}
\end{align}
The game value is
$
    V^{\ast}(x_{0})=x_{0}^{\top}\mathbf{P}^{\ast}x_{0}.
$

Equations \eqref{4}-\eqref{6} characterize the saddle point through the value matrix \(\textbf{P}^\ast\). To avoid treating \(\textbf{P}\) as an iterative variable, Section III recasts the GARE and stationarity conditions in the joint policy variable.

\section{Policy Game Riccati Equation}
\label{sec:PGRE}

We first introduce a joint policy variable that combines
the minimizing control policy and the maximizing disturbance
policy. This representation allows the GARE and the saddlepoint
stationarity conditions to be recast as an algebraic
equation whose unknowns are policy gains only.

For feedback gains
$\mathbf{K}\in\mathbb{R}^{m\times n}$ and
$\mathbf{L}\in\mathbb{R}^{q\times n}$, define
$
    \mathbf{\Theta}
    \triangleq\operatorname{col}(\mathbf{L},-\mathbf{K})
    \in\mathbb{R}^{(q+m)\times n},
$
together with
$
    \mathbf{B}_{g}
    \triangleq
    \begin{bmatrix}\mathbf{B}_{2}&\mathbf{B}_{1}\end{bmatrix}
    \in\mathbb{R}^{n\times(q+m)}
$
and
$
    \mathbf{R}_{g}
    \triangleq
    \operatorname{diag}
(-\gamma^{2}\mathbf{I}_{q},\mathbf{R}).
$
The saddle-point stationarity conditions
\eqref{5}-\eqref{6} become
\begin{equation}
    \mathbf{R}_{g}\mathbf{\Theta}
    +\mathbf{B}_{g}^{\top}\mathbf{P}=\mathbf{0},
\label{pgre_actorcritic}
\end{equation}
while the GARE takes the compact form
\begin{equation}
\begin{aligned}
    \mathbf{0}={}&\mathbf{A}^{\top}\mathbf{P}
    +\mathbf{P}\mathbf{A}+\mathbf{Q}
    -\mathbf{P}\mathbf{B}_{g}\mathbf{R}_{g}^{-1}
     \mathbf{B}_{g}^{\top}\mathbf{P}.
\end{aligned}
\label{pgre_gare_compact}
\end{equation}

The following identity rewrites the quadratic Riccati term as a
quadratic function of the joint policy matrix.

\begin{lemma}
\label{lem:game_sylvester_identity}
Let $\mathbf{P}=\mathbf{P}^{\top}$ and suppose that
$(\mathbf{P},\mathbf{\Theta})$ satisfies
$
    \mathbf{R}_{g}\mathbf{\Theta}
    +\mathbf{B}_{g}^{\top}\mathbf{P}=\mathbf{0}
$. Then \eqref{pgre_gare_compact} holds if and
only if
\begin{equation}
    \mathbf{A}^{\top}\mathbf{P}+\mathbf{P}\mathbf{A}
    =\mathbf{\Theta}^{\top}\mathbf{R}_{g}\mathbf{\Theta}
    -\mathbf{Q}.
\label{pgre_fixed_sylvester}
\end{equation}
\end{lemma}

\begin{proof}
Since $\mathbf{R}_{g}$ is symmetric and nonsingular,
we have
$
    \mathbf{\Theta}
    =-\mathbf{R}_{g}^{-1}\mathbf{B}_{g}^{\top}\mathbf{P}.
$
The symmetry of $\mathbf{P}$ and $\mathbf{R}_{g}$ then yields
$
\mathbf{\Theta}^{\top}\mathbf{R}_{g}\mathbf{\Theta}
    =\mathbf{P}\mathbf{B}_{g}\mathbf{R}_{g}^{-1}
     \mathbf{B}_{g}^{\top}\mathbf{P}.
$
Substitution into \eqref{pgre_gare_compact} proves
\eqref{pgre_fixed_sylvester}. Reversing the substitution proves the
converse.
\end{proof}

To recover $\mathbf{P}$ uniquely from a prescribed policy matrix, we
impose the following open-loop spectral condition.

\begin{assumption}
\label{ass:nonresonance_A}
The matrix $\mathbf{A}$ satisfies
\begin{equation}
    \operatorname{spec}(\mathbf{A})
    \cap\operatorname{spec}(-\mathbf{A})=\varnothing.
\label{pgre_nonresonance}
\end{equation}
\end{assumption}

Under Assumption~\ref{ass:nonresonance_A}, the Sylvester operator
$
    \mathcal{L}_{\mathbf{A}}(\mathbf{X})
    \triangleq\mathbf{A}^{\top}\mathbf{X}+\mathbf{X}\mathbf{A}
$
is nonsingular. Accordingly, for every
$\mathbf{\Theta}\in\mathbb{R}^{(q+m)\times n}$, define the
policy-dependent Sylvester lifting
\begin{equation}
\begin{aligned}
    \mathbf{P}_{\mathbf{A}}(\mathbf{\Theta})
    \triangleq\mathcal{L}_{\mathbf{A}}^{-1}\!\left(
    \mathbf{\Theta}^{\top}\mathbf{R}_{g}\mathbf{\Theta}
    -\mathbf{Q}\right).
\end{aligned}
\label{pgre_PA_definition}
\end{equation}
Equivalently, $\mathbf P_{\mathbf A}(\mathbf\Theta)$ is the unique
solution of
\begin{equation}
\begin{aligned}
    \mathbf{A}^{\top}\mathbf{P}_{\mathbf{A}}(\mathbf{\Theta})
    +\mathbf{P}_{\mathbf{A}}(\mathbf{\Theta})\mathbf{A}
    =\mathbf{\Theta}^{\top}\mathbf{R}_{g}\mathbf{\Theta}
    -\mathbf{Q}.
\end{aligned}
\label{pgre_PA_lyap}
\end{equation}
The right-hand side of \eqref{pgre_PA_lyap} is symmetric.
Transposing the equation therefore shows that
$\mathbf P_{\mathbf A}(\mathbf\Theta)^\top$ satisfies the same
Sylvester equation. Uniqueness then gives
$\mathbf{P}_{\mathbf{A}}(\mathbf{\Theta})
=\mathbf{P}_{\mathbf{A}}(\mathbf{\Theta})^{\top}$.
The vectorized representation of
\eqref{pgre_PA_lyap} is
\begin{equation}
\begin{aligned}
    \operatorname{vec}\!\left(
    \mathbf{P}_{\mathbf{A}}(\mathbf{\Theta})\right)
    ={}&\left[
    \mathbf{I}_{n}\otimes\mathbf{A}^{\top}
    +\mathbf{A}^{\top}\otimes\mathbf{I}_{n}
    \right]^{-1}\\
    &\times\operatorname{vec}\!\left(
    \mathbf{\Theta}^{\top}\mathbf{R}_{g}\mathbf{\Theta}
    -\mathbf{Q}\right).
\end{aligned}
\label{pgre_PA_vec}
\end{equation}

\begin{definition}
\label{def:PGRE}
Under Assumption~\ref{ass:nonresonance_A}, define the policy residual
\begin{equation}
    \mathbf{\Psi}_{\mathbf{A}}(\mathbf{\Theta})
    \triangleq\mathbf{R}_{g}\mathbf{\Theta}
    +\mathbf{B}_{g}^{\top}
     \mathbf{P}_{\mathbf{A}}(\mathbf{\Theta}).
\label{pgre_residual}
\end{equation}
The equation
$\mathbf{\Psi}_{\mathbf{A}}(\mathbf{\Theta})=\mathbf{0}$
is called the \emph{Policy Game Riccati Equation} (PGRE).
\end{definition}

In terms of $\mathbf{L}$ and $\mathbf{K}$, the PGRE $\mathbf{\Psi}_{\mathbf{A}}(\mathbf{\Theta})=\mathbf{0}$ becomes
\begin{equation}
\left\{
\begin{aligned}
    \mathbf{0}
    &=-\gamma^{2}\mathbf{L}
      +\mathbf{B}_{2}^{\top}
       \mathbf{P}_{\mathbf{A}}(\mathbf{\Theta}),\\
    \mathbf{0}
    &=-\mathbf{R}\mathbf{K}
      +\mathbf{B}_{1}^{\top}
       \mathbf{P}_{\mathbf{A}}(\mathbf{\Theta}).
\end{aligned}
\right.
\label{pgre_block_equations}
\end{equation}
Since $\mathcal{L}_{\mathbf{A}}^{-1}$ is linear,
$\mathbf{P}_{\mathbf{A}}(\mathbf{\Theta})$ is a matrix-valued
polynomial of degree at most two. Thus, the PGRE consists of
$(m+q)n$ scalar polynomial equations of degree at most two in the same
number of policy variables.
Define the symmetric GARE solution set and the PGRE solution set by
\begin{subequations}
\label{pgre_solution_sets}
\begin{align}
    \mathcal{P}_{\mathrm{G}}
    \triangleq\bigl\{\mathbf{P}\in\mathbb{R}^{n\times n}:{}&
    \mathbf{P}=\mathbf{P}^{\top},\
    \mathbf{P}\text{ satisfies \eqref{pgre_gare_compact}}\bigr\},
    \label{pgre_gare_solution_set}\\
    \mathcal{T}_{\mathrm{G}}
    \triangleq\bigl\{\mathbf{\Theta}
    \in\mathbb{R}^{(q+m)\times n}:{}&
    \mathbf{\Psi}_{\mathbf{A}}(\mathbf{\Theta})=\mathbf{0}\bigr\}.
    \label{pgre_policy_solution_set}
\end{align}
\end{subequations}

\begin{theorem}
\label{thm:PGRE_GARE_equivalence}
Under Assumption~\ref{ass:nonresonance_A}, consider the maps
$
    \mathfrak{T}:\mathcal{P}_{\mathrm{G}}
    \rightarrow\mathcal{T}_{\mathrm{G}}$,
    $
    \mathfrak{T}(\mathbf{P})
    \triangleq-\mathbf{R}_{g}^{-1}
      \mathbf{B}_{g}^{\top}\mathbf{P},
      $
      $
    \mathfrak{P}:\mathcal{T}_{\mathrm{G}}
    \rightarrow\mathcal{P}_{\mathrm{G}}$,
    $
    \mathfrak{P}(\mathbf{\Theta})
    \triangleq\mathbf{P}_{\mathbf{A}}(\mathbf{\Theta}).
$
The maps $\mathfrak T$ and $\mathfrak P$ are well-defined and mutually
inverse. Consequently, the PGRE solution set is in one-to-one
correspondence with the set of all symmetric GARE solutions.
\end{theorem}

\begin{proof}
Let $\mathbf{\Theta}\in\mathcal{T}_{\mathrm{G}}$ and set
$\mathbf{P}=\mathbf{P}_{\mathbf{A}}(\mathbf{\Theta})$. The preceding
symmetry argument gives $\mathbf{P}=\mathbf{P}^{\top}$.
Equations \eqref{pgre_PA_lyap} and \eqref{pgre_residual} yield
\eqref{pgre_fixed_sylvester} and \eqref{pgre_actorcritic},
respectively. Lemma~\ref{lem:game_sylvester_identity} then shows that
$\mathbf{P}$ satisfies \eqref{pgre_gare_compact}. Therefore,
$\mathfrak{P}(\mathbf{\Theta})\in\mathcal{P}_{\mathrm{G}}$, and
\eqref{pgre_actorcritic} gives
$
    \mathfrak{T}\bigl(\mathfrak{P}(\mathbf{\Theta})\bigr)
    =\mathbf{\Theta}.
$
Conversely, let $\mathbf{P}\in\mathcal{P}_{\mathrm{G}}$ and set
$\mathbf{\Theta}=\mathfrak{T}(\mathbf{P})$. By construction,
\eqref{pgre_actorcritic} holds. Lemma~\ref{lem:game_sylvester_identity}
then yields \eqref{pgre_fixed_sylvester}. The uniqueness of the
solution to \eqref{pgre_PA_lyap} implies
$\mathbf{P}=\mathbf{P}_{\mathbf{A}}(\mathbf{\Theta})$, which in turn
gives $\mathbf{\Psi}_{\mathbf{A}}(\mathbf{\Theta})=\mathbf{0}$.
Thus,
$
    \mathfrak{P}\bigl(\mathfrak{T}(\mathbf{P})\bigr)=\mathbf{P}.
$
The two identities establish the result.
\end{proof}

\begin{remark}
\label{rem:pgre_no_rank_condition}
Theorem~\ref{thm:PGRE_GARE_equivalence} requires neither a matching
condition between $\mathbf B_1$ and $\mathbf B_2$ nor injectivity of
$\mathbf B_g^\top$. The injectivity of $\mathfrak T$ on
$\mathcal P_{\mathrm G}$ follows from uniqueness of the solution to
the fixed Sylvester equation, not from a rank condition on
$\mathbf B_g^\top$.
\end{remark}

\begin{remark}
It should also be noted that the correspondence in Theorem 1 concerns all symmetric GARE solutions. A PGRE zero need not be stabilizing and therefore need not represent an admissible saddle-point policy for the infinite-horizon game. The game-relevant solution is selected by the closed-loop stability requirement together with the feasibility condition $\gamma>\gamma^{\ast}$.
\end{remark}

\section{Model-Based Critic-Free RL Algorithm}
\label{sec:model_based_rl}
 To remove the open-loop spectral requirement in Definition 1, we introduce a stable-anchor PGRE and derive Newton updates directly in policy space. We then establish their equivalence to offline SPUA
\cite{WuLuo2013}, thereby transferring its convergence guarantee under the same assumptions.

The PGRE in Definition~\ref{def:PGRE} relies on the Sylvester operator
associated with the open-loop matrix $\mathbf A$ and hence on
Assumption~\ref{ass:nonresonance_A}. To avoid this open-loop requirement in the Newton formulation, introduce a fixed
joint anchor policy
\begin{equation}
    \mathbf{\Theta}_{\mathrm{a}}
    \triangleq\operatorname{col}\!\left(
    \mathbf{L}_{\mathrm{a}},-\mathbf{K}_{\mathrm{a}}\right)
    \in\mathbb{R}^{s\times n},
    \qquad s\triangleq q+m,
\label{mb_anchor_policy}
\end{equation}
and let
$
    \mathbf{A}_{\mathrm{a}}
    \triangleq\mathbf{A}+\mathbf{B}_{g}\mathbf{\Theta}_{\mathrm{a}}.
$

\begin{assumption}
\label{ass:anchor_stable}
A fixed anchor $\mathbf{\Theta}_{\mathrm{a}}$ is available such that
$\mathbf{A}_{\mathrm{a}}$ is Hurwitz.
\end{assumption}

The anchor matrix $\mathbf A_{\mathrm a}$ induces the fixed Lyapunov
operator
$
    \mathcal{L}_{\mathrm{a}}(\mathbf{X})
    \triangleq\mathbf{A}_{\mathrm{a}}^{\top}\mathbf{X}
    +\mathbf{X}\mathbf{A}_{\mathrm{a}}.
$
Assumption~\ref{ass:anchor_stable} ensures that
$\mathcal L_{\mathrm a}$ is invertible. For any candidate policy
$\mathbf\Theta\in\mathbb R^{s\times n}$, define its deviation from the
anchor by
$
    \mathbf{\Delta}_{\mathrm{a}}(\mathbf{\Theta})
    \triangleq\mathbf{\Theta}-\mathbf{\Theta}_{\mathrm{a}},
$
and define
\begin{equation}
\begin{aligned}
    \mathbf{Q}_{\mathrm{a}}(\mathbf{\Theta})
    \triangleq{}&\mathbf{Q}
    +\mathbf{\Theta}_{\mathrm{a}}^{\top}
     \mathbf{R}_{g}\mathbf{\Theta}_{\mathrm{a}}-\mathbf{\Delta}_{\mathrm{a}}(\mathbf{\Theta})^{\top}
      \mathbf{R}_{g}
      \mathbf{\Delta}_{\mathrm{a}}(\mathbf{\Theta}).
\end{aligned}
\label{pgre_anchor_weight}
\end{equation}

Define the associated anchored lifting
$\mathbf P_{\mathrm a}(\mathbf\Theta)$ as the unique symmetric solution
to
\begin{equation}
\begin{aligned}
    \mathbf{0}={}&\mathbf{A}_{\mathrm{a}}^{\top}
    \mathbf{P}_{\mathrm{a}}(\mathbf{\Theta})
    +\mathbf{P}_{\mathrm{a}}(\mathbf{\Theta})
    \mathbf{A}_{\mathrm{a}}
    +\mathbf{Q}_{\mathrm{a}}(\mathbf{\Theta}),
\end{aligned}
\label{pgre_anchor_lift}
\end{equation}
or, equivalently,
$
    \mathbf{P}_{\mathrm{a}}(\mathbf{\Theta})
    =-\mathcal{L}_{\mathrm{a}}^{-1}\!\left(
    \mathbf{Q}_{\mathrm{a}}(\mathbf{\Theta})\right).
$
To eliminate it,
define the fixed anchor-response operator
\begin{equation}
\begin{aligned}
    \mathcal{G}_{\mathrm{a}}(\mathbf{X})
    \triangleq
    -\mathbf{B}_{g}^{\top}
    \mathcal{L}_{\mathrm{a}}^{-1}(\mathbf{X}),
    \qquad \mathbf{X}=\mathbf{X}^{\top}.
\end{aligned}
\label{mb_anchor_response_operator}
\end{equation}
It follows directly from \eqref{pgre_anchor_lift} that
\begin{equation}
    \mathcal{G}_{\mathrm{a}}\!\left(
    \mathbf{Q}_{\mathrm{a}}(\mathbf{\Theta})\right)
    =\mathbf{B}_{g}^{\top}
    \mathbf{P}_{\mathrm{a}}(\mathbf{\Theta}).
\label{mb_anchor_response_identity}
\end{equation}
We therefore define the anchored residual, without explicitly
retaining a policy-dependent value matrix, by
\begin{equation}
\begin{aligned}
    \mathbf{\Psi}_{\mathrm{a}}(\mathbf{\Theta})
    \triangleq{}&\mathbf{R}_{g}\mathbf{\Theta}
    +\mathcal{G}_{\mathrm{a}}\!\left(
    \mathbf{Q}_{\mathrm{a}}(\mathbf{\Theta})\right).
\end{aligned}
\label{pgre_anchor_residual}
\end{equation}
The equation
$
    \mathbf{\Psi}_{\mathrm{a}}(\mathbf{\Theta})=\mathbf{0}
$
is called the \emph{anchored PGRE}.
The correction in \eqref{pgre_anchor_weight} satisfies
\begin{equation}
\begin{aligned}
    -\mathbf{Q}_{\mathrm{a}}(\mathbf{\Theta})
    ={}&\mathbf{\Theta}^{\top}\mathbf{R}_{g}\mathbf{\Theta}
    -\mathbf{Q}
    -\mathbf{\Theta}^{\top}\mathbf{R}_{g}
      \mathbf{\Theta}_{\mathrm{a}}-\mathbf{\Theta}_{\mathrm{a}}^{\top}
      \mathbf{R}_{g}\mathbf{\Theta}.
\end{aligned}
\label{pgre_anchor_expand}
\end{equation}
The two cross terms compensate for replacing $\mathbf{A}$ with the
fixed stable matrix $\mathbf{A}_{\mathrm{a}}$.

\begin{proposition}
\label{prop:anchored_pgre_equivalence}
Under Assumption~\ref{ass:anchor_stable}, the solutions of the
anchored PGRE are in one-to-one correspondence with the symmetric
solutions of the GARE through
\begin{equation}
\left\{
\begin{aligned}
\mathbf{P}
&=-\mathcal{L}_{\mathrm{a}}^{-1}\!\left(
\mathbf{Q}_{\mathrm{a}}(\mathbf{\Theta})\right),\\
\mathbf{\Theta}
&=-\mathbf{R}_{g}^{-1}
\mathbf{B}_{g}^{\top}\mathbf{P}.
\end{aligned}
\right.
\label{pgre_anchor_maps}
\end{equation}
If Assumption~\ref{ass:nonresonance_A} also holds, the solution sets
of the anchored PGRE and the PGRE in Definition~\ref{def:PGRE}
coincide.
\end{proposition}

\begin{proof}
Let $\mathbf{\Psi}_{\mathrm{a}}(\mathbf{\Theta})=\mathbf{0}$ and
set
$\mathbf{P}=-\mathcal{L}_{\mathrm{a}}^{-1}
(\mathbf{Q}_{\mathrm{a}}(\mathbf{\Theta}))$. Equations
\eqref{mb_anchor_response_identity} and
\eqref{pgre_anchor_residual} give
$\mathbf{R}_{g}\mathbf{\Theta}
+\mathbf{B}_{g}^{\top}\mathbf{P}=\mathbf{0}$.
Substituting this identity and
\eqref{pgre_anchor_expand} into \eqref{pgre_anchor_lift} yields
\eqref{pgre_fixed_sylvester}. Lemma~\ref{lem:game_sylvester_identity}
therefore implies that $\mathbf{P}$ satisfies the GARE. Conversely,
let $\mathbf{P}$ be a symmetric GARE solution and set
$\mathbf{\Theta}=-\mathbf{R}_{g}^{-1}
\mathbf{B}_{g}^{\top}\mathbf{P}$. Reversing the preceding identities
shows that $\mathbf{P}$ satisfies \eqref{pgre_anchor_lift}. The
uniqueness of its solution gives
$\mathbf{P}=\mathbf{P}_{\mathrm{a}}(\mathbf{\Theta})$, and hence
$\mathbf{\Psi}_{\mathrm{a}}(\mathbf{\Theta})=\mathbf{0}$.
The final statement follows from
Theorem~\ref{thm:PGRE_GARE_equivalence}.
\end{proof}

For each joint policy $\mathbf\Theta$, define its closed-loop matrix by
$
    \mathbf{A}_{\mathbf{\Theta}}
    \triangleq\mathbf{A}+\mathbf{B}_{g}\mathbf{\Theta}.
$
The stabilizing policy set is
\begin{equation}
    \mathcal{S}_{g}
    \triangleq\left\{\mathbf{\Theta}\in\mathbb{R}^{s\times n}:
    \mathbf{A}_{\mathbf{\Theta}}\text{ is Hurwitz}\right\}.
\label{mb_stabilizing_set}
\end{equation}

The following residual identity links the anchored lifting to the
closed-loop policy-evaluation equation.
\begin{corollary}
\label{cor:anchored_pgre_identity}
Under Assumption~\ref{ass:anchor_stable}, every
$\mathbf{\Theta}\in\mathbb{R}^{s\times n}$ satisfies
\begin{equation}
\begin{aligned}
    &\mathbf{A}_{\mathbf{\Theta}}^{\top}
     \mathbf{P}_{\mathrm{a}}(\mathbf{\Theta})
    +\mathbf{P}_{\mathrm{a}}(\mathbf{\Theta})
     \mathbf{A}_{\mathbf{\Theta}}
    +\mathbf{Q}
    +\mathbf{\Theta}^{\top}\mathbf{R}_{g}\mathbf{\Theta}
    \\
    &\quad=
    \mathbf{\Delta}_{\mathrm{a}}(\mathbf{\Theta})^{\top}
    \mathbf{\Psi}_{\mathrm{a}}(\mathbf{\Theta})
    +\mathbf{\Psi}_{\mathrm{a}}(\mathbf{\Theta})^{\top}
    \mathbf{\Delta}_{\mathrm{a}}(\mathbf{\Theta}).
\end{aligned}
\label{mb_anchored_residual_identity}
\end{equation}
Consequently, at every zero of the anchored PGRE,
$\mathbf{P}_{\mathrm{a}}(\mathbf{\Theta})$ satisfies the closed-loop
policy-evaluation equation, while
$\mathbf{R}_{g}\mathbf{\Theta}
+\mathbf{B}_{g}^{\top}\mathbf{P}_{\mathrm{a}}(\mathbf{\Theta})
=\mathbf{0}$ supplies the stationarity condition.
\end{corollary}

\begin{proof}
Let $\mathbf{\Delta}_{\mathrm{a}}
=\mathbf{\Delta}_{\mathrm{a}}(\mathbf{\Theta})$ and
$\mathbf{P}_{\mathrm{a}}
=\mathbf{P}_{\mathrm{a}}(\mathbf{\Theta})$. Since
$\mathbf{A}_{\mathbf{\Theta}}
=\mathbf{A}_{\mathrm{a}}+\mathbf{B}_{g}\mathbf{\Delta}_{\mathrm{a}}$,
expand the left-hand side of
\eqref{mb_anchored_residual_identity}. Substitution of
\eqref{pgre_anchor_lift} and
$\mathbf{B}_{g}^{\top}\mathbf{P}_{\mathrm{a}}
=\mathbf{\Psi}_{\mathrm{a}}(\mathbf{\Theta})
-\mathbf{R}_{g}\mathbf{\Theta}$ leaves exactly the right-hand side
of \eqref{mb_anchored_residual_identity}. The final statement follows
by setting $\mathbf{\Psi}_{\mathrm{a}}(\mathbf{\Theta})=\mathbf{0}$.
\end{proof}

For $\mathbf{\Theta},\mathbf{E}\in\mathbb{R}^{s\times n}$, define
the symmetric cross term
\begin{equation}
\begin{aligned}
    \mathbf{S}_{\mathrm{a}}(\mathbf{\Theta},\mathbf{E})
    \triangleq{}&\mathbf{E}^{\top}\mathbf{R}_{g}
    \mathbf{\Delta}_{\mathrm{a}}(\mathbf{\Theta})+\mathbf{\Delta}_{\mathrm{a}}(\mathbf{\Theta})^{\top}
    \mathbf{R}_{g}\mathbf{E}.
\end{aligned}
\label{mb_anchor_cross_term}
\end{equation}

\begin{lemma}
\label{lem:pgre_derivative}
Under Assumption~\ref{ass:anchor_stable},
$\mathbf{\Psi}_{\mathrm{a}}$ is Fr\'echet differentiable, and
\begin{equation}
\begin{aligned}
    D\mathbf{\Psi}_{\mathrm{a}}(\mathbf{\Theta})[\mathbf{E}]
    ={}&\mathbf{R}_{g}\mathbf{E}
    -\mathcal{G}_{\mathrm{a}}\!\left(
    \mathbf{S}_{\mathrm{a}}(\mathbf{\Theta},\mathbf{E})\right).
\end{aligned}
\label{mb_DPsi}
\end{equation}
\end{lemma}

\begin{proof}
The definition of $\mathbf{Q}_{\mathrm{a}}$ gives the exact identity
\begin{equation*}
\begin{aligned}
    \mathbf{Q}_{\mathrm{a}}(\mathbf{\Theta}+\mathbf{E})
    ={}&\mathbf{Q}_{\mathrm{a}}(\mathbf{\Theta})
    -\mathbf{S}_{\mathrm{a}}(\mathbf{\Theta},\mathbf{E})
    -\mathbf{E}^{\top}\mathbf{R}_{g}\mathbf{E}.
\end{aligned}
\end{equation*}
Because $\mathcal{G}_{\mathrm{a}}$ is linear,
\eqref{pgre_anchor_residual} then yields
\begin{equation*}
\begin{aligned}
    &\mathbf{\Psi}_{\mathrm{a}}(\mathbf{\Theta}+\mathbf{E})
    -\mathbf{\Psi}_{\mathrm{a}}(\mathbf{\Theta})\\
    &\quad=\mathbf{R}_{g}\mathbf{E}
    -\mathcal{G}_{\mathrm{a}}\!\left(
    \mathbf{S}_{\mathrm{a}}(\mathbf{\Theta},\mathbf{E})\right)
    -\mathcal{G}_{\mathrm{a}}\!\left(
    \mathbf{E}^{\top}\mathbf{R}_{g}\mathbf{E}\right).
\end{aligned}
\end{equation*}
The first two terms on the right-hand side are linear in
$\mathbf{E}$, whereas the last term is of order
$\mathcal{O}(\|\mathbf{E}\|^{2})$. This proves
\eqref{mb_DPsi}.
\end{proof}

To derive an explicit actor-space representation, introduce
$
    \mathbf{M}_{\mathrm{a}}
    \triangleq\mathbf{A}_{\mathrm{a}}^{\top}
      \oplus\mathbf{A}_{\mathrm{a}}^{\top}.
$
Let $\mathbf{U}_{s,n}\in\mathbb{R}^{sn\times sn}$ denote the commutation
matrix characterized by
$
    \operatorname{vec}(\mathbf{E}^{\top})
    =\mathbf{U}_{s,n}\operatorname{vec}(\mathbf{E}).
$
The operator $\mathcal{G}_{\mathrm{a}}$ admits the matrix representation
$
    \mathbf{G}_{\mathrm{a}}
    \triangleq
    -\left(\mathbf{I}_{n}\otimes\mathbf{B}_{g}^{\top}\right)
    \mathbf{M}_{\mathrm{a}}^{-1}
    \in\mathbb{R}^{sn\times n^{2}},
$
such that
$
    \operatorname{vec}\!\left(
    \mathcal{G}_{\mathrm{a}}(\mathbf{X})\right)
    =\mathbf{G}_{\mathrm{a}}\operatorname{vec}(\mathbf{X}).
$
To represent $\mathbf S_{\mathrm a}(\mathbf\Theta,\mathbf E)$ in
vector form, define
\begin{equation}
\begin{aligned}
    \mathbf{H}_{\mathrm{a}}(\mathbf{\Theta})
    \triangleq{}&\mathbf{I}_{n}\otimes\left[
    \mathbf{\Delta}_{\mathrm{a}}(\mathbf{\Theta})^{\top}
    \mathbf{R}_{g}\right]\\
    &+\left[
    \mathbf{\Delta}_{\mathrm{a}}(\mathbf{\Theta})^{\top}
    \mathbf{R}_{g}\otimes\mathbf{I}_{n}\right]
    \mathbf{U}_{s,n}.
\end{aligned}
\label{mb_H_theta}
\end{equation}
Then, we have
$
    \operatorname{vec}\!\left(
    \mathbf{S}_{\mathrm{a}}(\mathbf{\Theta},\mathbf{E})\right)
    =\mathbf{H}_{\mathrm{a}}(\mathbf{\Theta})
    \operatorname{vec}(\mathbf{E}).
$
For $\boldsymbol{\theta}\triangleq\operatorname{vec}(\mathbf{\Theta})$,
the vectorized residual is
\begin{equation}
\begin{aligned}
    \mathbf{f}_{\mathrm{a}}(\mathbf{\Theta})
    \triangleq{}&\operatorname{vec}\!\left(
    \mathbf{\Psi}_{\mathrm{a}}(\mathbf{\Theta})\right)\\
    ={}&\left(\mathbf{I}_{n}\otimes\mathbf{R}_{g}\right)
    \boldsymbol{\theta}
    +\mathbf{G}_{\mathrm{a}}\operatorname{vec}\!\left(
    \mathbf{Q}_{\mathrm{a}}(\mathbf{\Theta})\right).
\end{aligned}
\label{mb_actor_residual_vec}
\end{equation}
Equation \eqref{mb_DPsi} gives
\begin{equation}
\begin{aligned}
    \operatorname{vec}\!\left(
    D\mathbf{\Psi}_{\mathrm{a}}(\mathbf{\Theta})[\mathbf{E}]
    \right)
    =\mathbf{J}_{\mathrm{a}}(\mathbf{\Theta})
    \operatorname{vec}(\mathbf{E}),
\end{aligned}
\label{mb_J_definition}
\end{equation}
where the Jacobian is
$
    \mathbf{J}_{\mathrm{a}}(\mathbf{\Theta})    \triangleq\mathbf{I}_{n}\otimes\mathbf{R}_{g}
    -\mathbf{G}_{\mathrm{a}}
    \mathbf{H}_{\mathrm{a}}(\mathbf{\Theta}).
$

\begin{lemma}
\label{thm:pgre_jacobian_nonsingular}
The actor-space Jacobian $\mathbf J_{\mathrm a}(\mathbf\Theta)$ is
nonsingular at every $\mathbf\Theta\in\mathcal S_g$.
\end{lemma}

\begin{proof}
Let $\mathbf{E}$ lie in the kernel of
$D\mathbf{\Psi}_{\mathrm{a}}(\mathbf{\Theta})$, and define
$
    \mathbf{\Pi}
    \triangleq\mathcal{L}_{\mathrm{a}}^{-1}\!\left(
    \mathbf{S}_{\mathrm{a}}(\mathbf{\Theta},\mathbf{E})\right).
$
The matrix $\mathbf{\Pi}$ is symmetric because the argument of
$\mathcal{L}_{\mathrm{a}}^{-1}$ is symmetric. From
\eqref{mb_anchor_response_operator} and \eqref{mb_DPsi},
\begin{equation}
    \mathbf{R}_{g}\mathbf{E}
    =-\mathbf{B}_{g}^{\top}\mathbf{\Pi},
    \qquad
    \mathbf{E}^{\top}\mathbf{R}_{g}
    =-\mathbf{\Pi}\mathbf{B}_{g}.
\label{mb_kernel_relation}
\end{equation}
Using \eqref{mb_kernel_relation} in
$\mathcal{L}_{\mathrm{a}}(\mathbf{\Pi})
=\mathbf{S}_{\mathrm{a}}(\mathbf{\Theta},\mathbf{E})$ yields
\begin{equation}
    \mathbf{A}_{\mathbf{\Theta}}^{\top}\mathbf{\Pi}
    +\mathbf{\Pi}\mathbf{A}_{\mathbf{\Theta}}=\mathbf{0}.
\label{mb_homogeneous_closed_loop}
\end{equation}
Because $\mathbf{A}_{\mathbf{\Theta}}$ is Hurwitz,
\eqref{mb_homogeneous_closed_loop} implies
$\mathbf{\Pi}=\mathbf{0}$. Equation \eqref{mb_kernel_relation} and
the nonsingularity of $\mathbf{R}_{g}$ then give
$\mathbf{E}=\mathbf{0}$. The derivative is therefore injective and,
because its domain and codomain both have dimension $sn$, nonsingular.
\end{proof}

\begin{remark}
\label{rem:indefinite_regularity}
The proof requires only that $\mathbf R_g$ be symmetric and
nonsingular. The positive definiteness is unnecessary. Thus, every
stabilizing policy is a regular point of the anchored PGRE even though
$\mathbf R_g$ is indefinite.
\end{remark}

\subsection{Critic-Free Policy Iteration Design}
\label{subsec:critic_free_newton}

At iteration $i$, given the current policy $\mathbf\Theta_i$, the
Newton direction $\mathbf E_i\in\mathbb R^{s\times n}$ is determined
by
\begin{equation}
    D\mathbf{\Psi}_{\mathrm{a}}(\mathbf{\Theta}_{i})
    [\mathbf{E}_{i}]
    =-\mathbf{\Psi}_{\mathrm{a}}(\mathbf{\Theta}_{i}),
\label{mb_newton_operator}
\end{equation}
and the full-step update is
\begin{equation}
    \mathbf{\Theta}_{i+1}
    =\mathbf{\Theta}_{i}+\mathbf{E}_{i}.
\label{mb_full_newton_update}
\end{equation}
If $\mathbf{\Theta}_{i}\in\mathcal{S}_{g}$,
Lemma~\ref{thm:pgre_jacobian_nonsingular} guarantees a unique
Newton direction. In vector form, $
    \mathbf{J}_{\mathrm{a}}(\mathbf{\Theta}_{i})
    \operatorname{vec}(\mathbf{E}_{i})
    =-\mathbf{f}_{\mathrm{a}}(\mathbf{\Theta}_{i}).
$
Only the $sn$ entries of the policy increment are unknown. The fixed operator $\mathbf{G}_{\mathrm{a}}$ supplies
the required value sensitivity without constructing a
policy-dependent critic matrix.

The following theorem shows that the critic-free Newton step is
exactly equivalent to a simultaneous PI step.

\begin{theorem}
\label{thm:pgre_newton_equivalence}
Let $\mathbf{\Theta}_{i}\in\mathcal{S}_{g}$,
$\mathbf{E}_{i}$ satisfy \eqref{mb_newton_operator}, and define
$\mathbf{\Theta}_{i+1}$ by \eqref{mb_full_newton_update}. Introduce
the auxiliary matrix
\begin{equation}
    \widehat{\mathbf P}_{i+1}
    \triangleq\mathcal{L}_{\mathrm{a}}^{-1}\!\left(
    \mathbf{S}_{\mathrm{a}}(\mathbf{\Theta}_{i},\mathbf{E}_{i})
    -\mathbf{Q}_{\mathrm{a}}(\mathbf{\Theta}_{i})\right).
\label{mb_Pi_newton_lift}
\end{equation}
Then $\widehat{\mathbf P}_{i+1}
=\widehat{\mathbf P}_{i+1}^{\top}$ is the unique solution of
\begin{equation}
\begin{aligned}
    \mathbf{0}={}&
    \mathbf{A}_{\mathbf{\Theta}_{i}}^{\top}
    \widehat{\mathbf P}_{i+1}
    +\widehat{\mathbf P}_{i+1}
     \mathbf{A}_{\mathbf{\Theta}_{i}}
    +\mathbf{Q}
    +\mathbf{\Theta}_{i}^{\top}
     \mathbf{R}_{g}\mathbf{\Theta}_{i}.
\end{aligned}
\label{mb_compact_policy_evaluation}
\end{equation}
Moreover, the Newton update satisfies
\begin{equation}
    \mathbf{\Theta}_{i+1}
    =-\mathbf{R}_{g}^{-1}\mathbf{B}_{g}^{\top}
      \widehat{\mathbf P}_{i+1}.
\label{mb_compact_policy_improvement}
\end{equation}
Conversely, if $\widehat{\mathbf P}_{i+1}$ solves
\eqref{mb_compact_policy_evaluation} and
$\mathbf{\Theta}_{i+1}$ is defined by
\eqref{mb_compact_policy_improvement}, then
$\mathbf{E}_{i}=\mathbf{\Theta}_{i+1}-\mathbf{\Theta}_{i}$ is the
unique Newton direction in \eqref{mb_newton_operator}.
\end{theorem}

\begin{proof}
Let
$\mathbf{S}_{i}
=\mathbf{S}_{\mathrm{a}}(\mathbf{\Theta}_{i},\mathbf{E}_{i})$
and
$\mathbf{Q}_{\mathrm{a},i}
=\mathbf{Q}_{\mathrm{a}}(\mathbf{\Theta}_{i})$.
Using \eqref{mb_anchor_response_operator}, the Newton equation is
equivalent to
\begin{equation*}
\begin{aligned}
    \mathbf{0}
    ={}&\mathbf{R}_{g}(\mathbf{\Theta}_{i}+\mathbf{E}_{i})
    +\mathcal{G}_{\mathrm{a}}
      (\mathbf{Q}_{\mathrm{a},i}-\mathbf{S}_{i})\\
    ={}&\mathbf{R}_{g}\mathbf{\Theta}_{i+1}
    +\mathbf{B}_{g}^{\top}\widehat{\mathbf P}_{i+1},
\end{aligned}
\end{equation*}
which proves \eqref{mb_compact_policy_improvement}. Moreover,
\eqref{mb_Pi_newton_lift} gives
\begin{equation}
    \mathcal{L}_{\mathrm{a}}(\widehat{\mathbf P}_{i+1})
    =\mathbf{S}_{i}-\mathbf{Q}_{\mathrm{a},i}.
\label{mb_sum_sylvester}
\end{equation}
Let
$\mathbf{\Delta}_{i}
=\mathbf{\Delta}_{\mathrm{a}}(\mathbf{\Theta}_{i})$.
Using
$\mathbf{A}_{\mathbf{\Theta}_{i}}
=\mathbf{A}_{\mathrm{a}}+\mathbf{B}_{g}\mathbf{\Delta}_{i}$,
\eqref{mb_compact_policy_improvement}, and
\eqref{mb_sum_sylvester}, we obtain
\begin{equation*}
\begin{aligned}
    &\mathbf{A}_{\mathbf{\Theta}_{i}}^{\top}
       \widehat{\mathbf P}_{i+1}
    +\widehat{\mathbf P}_{i+1}
       \mathbf{A}_{\mathbf{\Theta}_{i}}
    +\mathbf{Q}
    +\mathbf{\Theta}_{i}^{\top}\mathbf{R}_{g}\mathbf{\Theta}_{i}\\
    &\quad=\mathbf{S}_{i}-\mathbf{Q}_{\mathrm{a},i}
    +\mathbf{Q}
    +\mathbf{\Theta}_{i}^{\top}\mathbf{R}_{g}\mathbf{\Theta}_{i}\\
    &\qquad
    -\mathbf{\Delta}_{i}^{\top}\mathbf{R}_{g}
      \mathbf{\Theta}_{i+1}
    -\mathbf{\Theta}_{i+1}^{\top}\mathbf{R}_{g}
      \mathbf{\Delta}_{i}
    =\mathbf{0}.
\end{aligned}
\end{equation*}
The last equality follows from \eqref{pgre_anchor_weight},
$\mathbf{\Theta}_{i+1}=\mathbf{\Theta}_{i}+\mathbf{E}_{i}$,
and the definition of $\mathbf{S}_{i}$. This proves
\eqref{mb_compact_policy_evaluation}. Its symmetric solution is
unique because $\mathbf{A}_{\mathbf{\Theta}_{i}}$ is Hurwitz.

Conversely, suppose that \eqref{mb_compact_policy_evaluation} and
\eqref{mb_compact_policy_improvement} hold, and set
$\mathbf{E}_{i}=\mathbf{\Theta}_{i+1}-\mathbf{\Theta}_{i}$.
Reversing the preceding identities yields
\eqref{mb_sum_sylvester} and hence
\eqref{mb_Pi_newton_lift}. Substitution into
\eqref{mb_compact_policy_improvement} recovers
\eqref{mb_newton_operator}. Uniqueness follows from
Lemma~\ref{thm:pgre_jacobian_nonsingular}.
\end{proof}

\begin{corollary}
\label{cor:pgre_spua_sequence_equivalence}
Let $\mathbf{P}_{i}$ be generated by the offline
SPUA in \cite[Algorithm~4]{WuLuo2013}, with
$\mathbf{P}_{0}\in\mathbb{P}_0$ as specified in
\cite[Lemma~3]{WuLuo2013}. Set
\begin{equation}
    \mathbf{\Theta}_{\mathrm a}
    =\mathbf{\Theta}_0
    =-\mathbf{R}_g^{-1}\mathbf{B}_g^\top\mathbf{P}_{0}.
\label{mb_consistent_initialization}
\end{equation}
Then Algorithm~\ref{alg:critic_free_pgre} generates
the same policy sequence as the offline SPUA, namely,
\begin{equation}
    \mathbf{\Theta}_i
    =-\mathbf{R}_g^{-1}\mathbf{B}_g^\top\mathbf{P}_{i}.
\label{mb_spua_sequence_identity}
\end{equation}
Consequently, $\mathbf{\Theta}_i\to\mathbf{\Theta}^\ast$ and
$(\mathbf{L}_i,\mathbf{K}_i)\to
(\mathbf{L}^\ast,\mathbf{K}^\ast)$.
\end{corollary}

\begin{proof}
Equation \eqref{mb_spua_sequence_identity} holds at $i=0$ by
\eqref{mb_consistent_initialization}. If it holds at iteration $i$,
the Lyapunov equation of the offline SPUA coincides with
\eqref{mb_compact_policy_evaluation}, and its policy update coincides
with \eqref{mb_compact_policy_improvement}. Theorem
\ref{thm:pgre_newton_equivalence} therefore proves the identity at
$i+1$. Induction establishes \eqref{mb_spua_sequence_identity} for
all $i$. The convergence conclusion follows from
\cite[Theorem~2]{WuLuo2013}.
\end{proof}

The model-based implementation is summarized in
Algorithm~\ref{alg:critic_free_pgre}.

\begin{algorithm}[!t]
\caption{Model-Based Critic-Free PI}
\label{alg:critic_free_pgre}
\renewcommand{\algorithmicrequire}{\textbf{Input:}}
\renewcommand{\algorithmicensure}{\textbf{Output:}}
\begin{algorithmic}[1]
\REQUIRE Stabilizing anchor $\mathbf{\Theta}_{\mathrm{a}}=\mathbf{\Theta}_{0}\in\mathcal{S}_{g}$; tolerance
$\varepsilon>0$.
\ENSURE Optimal gains $\mathbf L^{\ast}$ and $\mathbf K^{\ast}$.
\STATE Set $i\gets0$.
\WHILE{$\|\mathbf{\Theta}_{i+1}-\mathbf{\Theta}_{i}\|>\varepsilon$}
    \STATE Solve
    $\mathbf{J}_{\mathrm{a}}(\mathbf{\Theta}_{i})\mathbf{e}_{i}
    =-\mathbf{f}_{\mathrm{a}}(\mathbf{\Theta}_{i})$.
    \STATE Set $\mathbf{\Theta}_{i+1}\gets\mathbf{\Theta}_{i}
    +\operatorname{unvec}_{s\times n}(\mathbf{e}_{i})$.
    \STATE Set $i\gets i+1$.
\ENDWHILE
\STATE Recover $\mathbf{L}_{i}\gets
\mathbf{\Theta}_{i}(1\!:\!q,:)$.
\STATE Recover $\mathbf{K}_{i}\gets
-\mathbf{\Theta}_{i}(q+1\!:\!q+m,:)$.
\end{algorithmic}
\end{algorithm}

\section{Data-Driven Critic-Free RL Algorithm}
\label{sec:data_based_critic_free}

The model-based method in Section~\ref{sec:model_based_rl} requires
the system matrices to construct the anchor-response operator. We
remove this requirement using one batch of state, control, and
disturbance data. Projecting an integral system identity onto the
nullspace of the endpoint increments eliminates the value matrix and
yields an actor-only policy residual.

Let the stabilizing anchor
$\mathbf{\Theta}_{\mathrm a}
=\operatorname{col}(\mathbf L_{\mathrm a},-\mathbf K_{\mathrm a})$
satisfy Assumption~\ref{ass:anchor_stable}, and define
\begin{equation}
    \mathbf v_{\mathrm a}(t)
    \triangleq
    \operatorname{col}\!\left(
    \mathbf w(t)-\mathbf L_{\mathrm a}\mathbf x(t),
    \mathbf u(t)+\mathbf K_{\mathrm a}\mathbf x(t)\right).
\label{db_anchor_behavior}
\end{equation}
With this definition, the system dynamics become
\begin{equation}
    \dot{\mathbf x}(t)
    =\mathbf A_{\mathrm a}\mathbf x(t)
    +\mathbf B_g\mathbf v_{\mathrm a}(t).
\label{db_anchor_behavior_dynamics}
\end{equation}

Consider one trajectory over
$
    0=t_0<t_1<\cdots<t_N$,
$
    \mathcal I_j\triangleq[t_{j-1},t_j].
$
For each $\mathcal I_j$, define
\begin{subequations}
\label{db_raw_interval_data}
\begin{align}
    \mathbf X_j
    &\triangleq
    \int_{\mathcal I_j}\mathbf x(\tau)\mathbf x(\tau)^\top\,d\tau,
    \label{db_Xj}\\
    \mathbf S_{w,j}
    &\triangleq
    \int_{\mathcal I_j}\mathbf w(\tau)\mathbf x(\tau)^\top\,d\tau,
    \label{db_Swj}\\
    \mathbf S_{u,j}
    &\triangleq
    \int_{\mathcal I_j}\mathbf u(\tau)\mathbf x(\tau)^\top\,d\tau,
    \label{db_Suj}\\
    \mathbf D_j
    &\triangleq
    \mathbf x(t_j)\mathbf x(t_j)^\top
    -\mathbf x(t_{j-1})\mathbf x(t_{j-1})^\top.
    \label{db_Dj}
\end{align}
\end{subequations}
The corresponding anchor-deviation data are
\begin{equation}
\begin{aligned}
    \mathbf Z_{\mathrm a,j}
    &\triangleq
    \int_{\mathcal I_j}
    \mathbf v_{\mathrm a}(\tau)\mathbf x(\tau)^\top\,d\tau\\
    &=\operatorname{col}\!\left(
    \mathbf S_{w,j}-\mathbf L_{\mathrm a}\mathbf X_j,\,
    \mathbf S_{u,j}+\mathbf K_{\mathrm a}\mathbf X_j\right)
    \in\mathbb R^{s\times n}.
\end{aligned}
\label{db_Zaj}
\end{equation}

Let $d_n\triangleq n(n+1)/2$ and define
$
    \mathcal D_x
    \triangleq
    \left[
    \operatorname{vech}(\mathbf D_1)\ \cdots\
    \operatorname{vech}(\mathbf D_N)
    \right]
    \in\mathbb R^{d_n\times N}.
$
Let
$\mathbf N_D=[\alpha_{j\ell}]\in\mathbb R^{N\times\nu}$,
where
$\nu\triangleq N-\operatorname{rank}(\mathcal D_x)$,
have orthonormal columns spanning $\ker(\mathcal D_x)$. Thus,
\begin{equation}
    \sum_{j=1}^N\alpha_{j\ell}\mathbf D_j=\mathbf 0,
    \qquad \ell=1,\ldots,\nu.
\label{db_endpoint_cancellation}
\end{equation}
Using the same coefficients, stack the projected behavior data as
\begin{equation}
    \mathbf Z_{\mathrm a}
    \triangleq
    \begin{bmatrix}
      \operatorname{vec}(\sum_{j=1}^{N}\alpha_{j1}
      \mathbf{Z}_{\mathrm{a},j})^\top\\
      \vdots\\
      \operatorname{vec}(\sum_{j=1}^{N}\alpha_{j\nu}
      \mathbf{Z}_{\mathrm{a},j})^\top
    \end{bmatrix}
    \in\mathbb R^{\nu\times sn}.
\label{db_projected_behavior_matrix}
\end{equation}

\begin{assumption}
\label{ass:data_rank}
The data satisfy
\begin{equation}
    \operatorname{rank}(\mathbf Z_{\mathrm a})=sn.
\label{db_rank_condition}
\end{equation}
\end{assumption}

Assumption~\ref{ass:data_rank} requires $\nu\geq sn$. Since
$\operatorname{rank}(\mathcal D_x)\leq d_n$, the sample-count
condition $N\geq d_n+sn$ ensures this dimension requirement but not
\eqref{db_rank_condition}. Full column rank also requires excitation
in both components of $\mathbf v_{\mathrm a}(t)$.

For a candidate policy
$\mathbf\Theta\in\mathbb R^{s\times n}$, let
$\mathbf P=\mathbf P_{\mathrm a}(\mathbf\Theta)$ and
$\mathbf\Psi=\mathbf\Psi_{\mathrm a}(\mathbf\Theta)$. Along
\eqref{db_anchor_behavior_dynamics},
\eqref{pgre_anchor_lift} and \eqref{pgre_anchor_residual} give
\begin{equation}
\begin{aligned}
    \frac{d}{dt}
    \left(\mathbf x^\top\mathbf P\mathbf x\right)
    ={}&
    -\mathbf x^\top\mathbf Q_{\mathrm a}(\mathbf\Theta)\mathbf x+2\mathbf v_{\mathrm a}^\top
    \left[\mathbf\Psi-\mathbf R_g\mathbf\Theta\right]\mathbf x.
\end{aligned}
\label{db_value_derivative}
\end{equation}
Integrating over $\mathcal I_j$ yields
\begin{equation}
\begin{aligned}
    &\operatorname{tr}\!\left(
    \mathbf Q_{\mathrm a}(\mathbf\Theta)\mathbf X_j\right)
    +2\operatorname{tr}\!\left(
    \mathbf Z_{\mathrm a,j}^\top\mathbf R_g\mathbf\Theta\right)\\
    &\quad=
    -\operatorname{tr}(\mathbf P\mathbf D_j)
    +2\operatorname{tr}\!\left(
    \mathbf Z_{\mathrm a,j}^\top\mathbf\Psi\right).
\end{aligned}
\label{db_interval_identity}
\end{equation}
Define the data residual
$\mathbf\Phi_{\mathrm a}:\mathbb R^{s\times n}\to\mathbb R^\nu$
componentwise by
\begin{equation}
\begin{aligned}
    \bigl[\mathbf\Phi_{\mathrm a}(\mathbf\Theta)\bigr]_\ell
    \triangleq{}&
    \operatorname{tr}(
    \mathbf Q_{\mathrm a}(\mathbf\Theta)
    \sum_{j=1}^{N}\alpha_{j\ell}\mathbf X_j)\\
    &+2\operatorname{tr}(
    (\sum_{j=1}^{N}\alpha_{j\ell}
    \mathbf Z_{\mathrm a,j})^\top
    \mathbf R_g\mathbf\Theta).
\end{aligned}
\label{db_data_residual}
\end{equation}

\begin{lemma}
\label{thm:data_residual_equivalence}
Under Assumption~\ref{ass:anchor_stable},
\begin{equation}
    \mathbf{\Phi}_{\mathrm a}(\mathbf{\Theta})
    =
    2\mathbf Z_{\mathrm a}
    \operatorname{vec}\!\left(
    \mathbf{\Psi}_{\mathrm a}(\mathbf{\Theta})\right)
\label{dd_residual_relation}
\end{equation}
for every
$\mathbf{\Theta}\in\mathbb R^{s\times n}$.
If Assumption~\ref{ass:data_rank} also holds, then
$\mathbf{\Phi}_{\mathrm a}$ and
$\mathbf{\Psi}_{\mathrm a}$ have the same zero set.
Consequently, the zeros of the data residual are in one-to-one
correspondence with the symmetric GARE solutions through
Proposition~\ref{prop:anchored_pgre_equivalence}.
\end{lemma}

\begin{proof}
Multiply \eqref{db_interval_identity} by $\alpha_{j\ell}$ and sum
over $j$. Equation \eqref{db_endpoint_cancellation} cancels the
endpoint term and gives
\begin{equation*}
    \bigl[\mathbf\Phi_{\mathrm a}(\mathbf\Theta)\bigr]_\ell
    =
    2\operatorname{tr}((\sum_{j=1}^{N}\alpha_{j\ell}
      \mathbf{Z}_{\mathrm{a},j})^{\top}
    \mathbf\Psi_{\mathrm a}(\mathbf\Theta)).
\end{equation*}
Stacking these identities and using the trace--vectorization identity
yields \eqref{dd_residual_relation}.
Suppose now that
$\mathbf\Phi_{\mathrm a}(\mathbf\Theta)=\mathbf0$.
Under Assumption~\ref{ass:data_rank},
$\mathbf Z_{\mathrm a}$ has full column rank. Therefore,
\eqref{dd_residual_relation} implies
$\mathbf\Psi_{\mathrm a}(\mathbf\Theta)=\mathbf0$.
Conversely,
$\mathbf\Psi_{\mathrm a}(\mathbf\Theta)=\mathbf0$
immediately gives
$\mathbf\Phi_{\mathrm a}(\mathbf\Theta)=\mathbf0$
from \eqref{dd_residual_relation}.
Thus, $\mathbf\Phi_{\mathrm a}$ and
$\mathbf\Psi_{\mathrm a}$ have the same zero set.
The final claim follows from
Proposition~\ref{prop:anchored_pgre_equivalence}.
\end{proof}

\begin{remark}
The orthonormal basis $\mathbf N_D$ of $\ker(\mathcal D_x)$ is not
unique. If
$\widetilde{\mathbf N}_D=\mathbf N_D\mathbf U$
for any orthogonal $\mathbf U$, the projected quantities are changed only
by the left orthogonal transformation. Consequently, the
rank condition in Assumption~\ref{ass:data_rank}, the zero set of the
projected residual, and the uniquely recovered actor are independent of
the particular nullspace basis. Thus, actor informativity is a property
of the projected data subspace rather than of a specific numerical basis.
\end{remark}

Let
$\mathbf\Delta_i\triangleq
\mathbf\Theta_i-\mathbf\Theta_{\mathrm a}$.
For $\mathbf E\in\mathbb R^{s\times n}$,
$
    D\mathbf Q_{\mathrm a}(\mathbf\Theta_i)[\mathbf E]
    =
    -\mathbf E^\top\mathbf R_g\mathbf\Delta_i
    -\mathbf\Delta_i^\top\mathbf R_g\mathbf E.
$
Because $\sum_{j=1}^{N}\alpha_{j\ell}\mathbf X_j$ is symmetric,
\begin{equation}
\begin{aligned}
    D\bigl[\mathbf\Phi_{\mathrm a}(\mathbf\Theta_i)\bigr]_\ell[\mathbf E]
    =
    2\operatorname{tr}(
    [
    \sum_{j=1}^{N}\alpha_{j\ell}\mathbf Z_{\mathrm a,j}
    \!\!-\mathbf\Delta_i
    \!\!\sum_{j=1}^{N}\alpha_{j\ell}\mathbf X_j
    ]^\top
    \mathbf R_g\mathbf E).
\end{aligned}
\label{db_DPhi_component}
\end{equation}
Define
$\mathbf\Gamma_i\in\mathbb R^{\nu\times sn}$ by
\begin{equation}
    \bigl[\mathbf\Gamma_i\bigr]_{\ell,:}
    \triangleq
    2\operatorname{vec}(
    \mathbf R_g[
    \sum_{j=1}^{N}\alpha_{j\ell}\mathbf Z_{\mathrm a,j}
    -\mathbf\Delta_i
    \sum_{j=1}^{N}\alpha_{j\ell}\mathbf X_j
    ])^\top.
\label{db_Gamma}
\end{equation}
Then
$
    D\mathbf\Phi_{\mathrm a}(\mathbf\Theta_i)[\mathbf E]
    =\mathbf\Gamma_i\operatorname{vec}(\mathbf E).
$
The actor-only Newton direction
$\mathbf e_i=\operatorname{vec}(\mathbf\Theta_{i+1}-\mathbf\Theta_i)$
satisfies
\begin{equation}
    \mathbf\Gamma_i\mathbf e_i
    =-\mathbf\Phi_{\mathrm a}(\mathbf\Theta_i).
\label{db_actor_newton}
\end{equation}
This system is consistent even when $\nu>sn$.
The next policy solves
\begin{equation}
    \mathbf\Gamma_i\operatorname{vec}(\mathbf\Theta_{i+1})
    =\mathbf b_i,
\label{db_direct_theta}
\end{equation}
where
$
    \bigl[\mathbf b_i\bigr]_\ell
    \triangleq
    -\operatorname{tr}(
    [\mathbf Q+\mathbf\Theta_i^\top\mathbf R_g\mathbf\Theta_i]
    \sum_{j=1}^{N}\alpha_{j\ell}\mathbf X_j).
$
Equation \eqref{db_direct_theta} follows by substituting
$\mathbf e_i=\operatorname{vec}(\mathbf\Theta_{i+1}-\mathbf\Theta_i)$
into \eqref{db_actor_newton} and using the symmetry of
$\mathbf R_g$ and
$\sum_{j=1}^{N}\alpha_{j\ell}\mathbf X_j$.
For implementation, write
$\mathbf\Theta_i=\operatorname{col}(\mathbf L_i,-\mathbf K_i)$ and
define
\begin{subequations}
\label{db_iteration_data}
\begin{align}
    \mathbf M_{w,i,\ell}
    &\triangleq
    \sum_{j=1}^{N}\alpha_{j\ell}\mathbf S_{w,j}
    -\mathbf L_i
    \sum_{j=1}^{N}\alpha_{j\ell}\mathbf X_j,
    \label{db_Mw}\\
    \mathbf M_{u,i,\ell}
    &\triangleq
    \sum_{j=1}^{N}\alpha_{j\ell}\mathbf S_{u,j}
    +\mathbf K_i
    \sum_{j=1}^{N}\alpha_{j\ell}\mathbf X_j,
    \label{db_Mu}\\
    \mathbf Q_i
    &\triangleq
    \mathbf Q+\mathbf K_i^\top\mathbf R\mathbf K_i
    -\gamma^2\mathbf L_i^\top\mathbf L_i.
    \label{db_Qi}
\end{align}
\end{subequations}
Then \eqref{db_direct_theta} is equivalent to
\begin{equation}
\begin{aligned}
    &2\gamma^2\operatorname{tr}\!\left(
    \mathbf M_{w,i,\ell}^\top\mathbf L_{i+1}\right)
    +2\operatorname{tr}\!\left(
    \mathbf M_{u,i,\ell}^\top\mathbf R\mathbf K_{i+1}\right)\\
    &\qquad=
    \operatorname{tr}(
    \mathbf Q_i
    \sum_{j=1}^{N}\alpha_{j\ell}\mathbf X_j),
    \qquad \ell=1,\ldots,\nu.
\end{aligned}
\label{db_gain_trace_equation}
\end{equation}
Define
$
    \boldsymbol\eta_{i+1}
    \triangleq
    \operatorname{col}\!\left(
    \operatorname{vec}(\mathbf L_{i+1}),
    \operatorname{vec}(\mathbf K_{i+1})\right)
    \in\mathbb R^{sn},
$
and construct
$\mathbf\Lambda_i\in\mathbb R^{\nu\times sn}$ and
$\mathbf y_i\in\mathbb R^\nu$ by
\begin{subequations}
\label{db_Lambda_y}
\begin{align}
    \bigl[\mathbf\Lambda_i\bigr]_{\ell,:}
    &\triangleq
    2
    \begin{bmatrix}
      \gamma^2\operatorname{vec}(\mathbf M_{w,i,\ell})^\top&
      \operatorname{vec}(\mathbf R\mathbf M_{u,i,\ell})^\top
    \end{bmatrix},
    \label{db_Lambda}\\
    \bigl[\mathbf y_i\bigr]_\ell
    &\triangleq
    \operatorname{tr}(
    \mathbf Q_i
    \sum_{j=1}^{N}\alpha_{j\ell}\mathbf X_j).
    \label{db_yi}
\end{align}
\end{subequations}
The gain update is
\begin{equation}
    \mathbf\Lambda_i\boldsymbol\eta_{i+1}=\mathbf y_i.
\label{db_gain_regression}
\end{equation}

\begin{lemma}
\label{thm:data_update_equivalence}
Suppose that Assumptions~\ref{ass:anchor_stable} and
\ref{ass:data_rank} hold and
$\mathbf\Theta_i\in\mathcal S_g$. Then
$\mathbf\Gamma_i$ and $\mathbf\Lambda_i$ have full column rank.
Equations \eqref{db_actor_newton}, \eqref{db_direct_theta}, and
\eqref{db_gain_regression} have unique, mutually equivalent
solutions and produce the model-based Newton policy
$\mathbf\Theta_{i+1}$.
\end{lemma}

\begin{proof}
Differentiating \eqref{dd_residual_relation} gives
$
    D\mathbf\Phi_{\mathrm a}(\mathbf\Theta_i)
    =
    2\mathbf Z_{\mathrm a}
    D\mathbf\Psi_{\mathrm a}(\mathbf\Theta_i)
$
after vectorization. Lemma~\ref{thm:pgre_jacobian_nonsingular}
and Assumption~\ref{ass:data_rank} therefore imply
$\operatorname{rank}(\mathbf\Gamma_i)=sn$. Moreover,
\eqref{db_actor_newton} reduces to
$
    D\mathbf\Psi_{\mathrm a}(\mathbf\Theta_i)[\mathbf E_i]
    =-\mathbf\Psi_{\mathrm a}(\mathbf\Theta_i),
$
so it produces the model-based Newton update. Equation
\eqref{db_direct_theta} is an algebraic rearrangement of
\eqref{db_actor_newton}. Finally,
$\operatorname{vec}(\mathbf\Theta_{i+1})$ and
$\boldsymbol\eta_{i+1}$ are related by a fixed signed-permutation
matrix. Thus, \eqref{db_gain_regression} is a bijective change of
coordinates of \eqref{db_direct_theta}, which gives
$\operatorname{rank}(\mathbf\Lambda_i)=sn$ and proves the claim.
\end{proof}

\begin{algorithm}[H]
\caption{Data-Driven Critic-Free PI}
\label{alg:data_critic_free_pgre}
\renewcommand{\algorithmicrequire}{\textbf{Input:}}
\renewcommand{\algorithmicensure}{\textbf{Output:}}
\begin{algorithmic}[1]
\REQUIRE Stabilizing anchor
$(\mathbf L_{\mathrm a},\mathbf K_{\mathrm a})$;
$\{\mathbf x(t),\mathbf u(t),\mathbf w(t)\}_{t=0}^{t_N}$;
tolerance $\varepsilon>0$.
\ENSURE Optimal gains $\mathbf L^{\ast}$ and $\mathbf K^{\ast}$.
\STATE Compute the interval data in \eqref{db_raw_interval_data}.
\STATE Form $\mathcal D_x$ and compute an orthonormal basis
$\mathbf N_D$ of $\ker(\mathcal D_x)$.
\STATE Set $i\gets0$ and
$\mathbf\Theta_0\gets\operatorname{col}(\mathbf L_0,-\mathbf K_0)$.
\WHILE{$\|\mathbf\Theta_{i+1}-\mathbf\Theta_{i}\|>\varepsilon$}
    \STATE Compute $\boldsymbol\eta_{i+1}$ from
    \eqref{db_gain_regression}.
    \STATE Recover $\mathbf L_{i+1}$ and $\mathbf K_{i+1}$, and set
    $\mathbf\Theta_{i+1}\gets
    \operatorname{col}(\mathbf L_{i+1},-\mathbf K_{i+1})$.
    \STATE Set $i\gets i+1$.
\ENDWHILE
\end{algorithmic}
\end{algorithm}

\section{Data Informativity Analysis}
\label{sec:data_informativity}

This section characterizes when the collected data uniquely determine
the policy update without identifying the value matrix. 

Integrating \eqref{mb_compact_policy_evaluation} over
$\mathcal I_j$ and using \eqref{mb_compact_policy_improvement} gives
\begin{equation}
\begin{aligned}
    &2\gamma^2\operatorname{tr}\!\left[
      (\mathbf S_{w,j}-\mathbf L_i\mathbf X_j)^{\top}
      \mathbf L_{i+1}\right]\\
    &\qquad
      +2\operatorname{tr}\!\left[
      (\mathbf S_{u,j}+\mathbf K_i\mathbf X_j)^{\top}
      \mathbf R\mathbf K_{i+1}\right]\\
    &\qquad=
      \operatorname{tr}(\mathbf Q_i\mathbf X_j)
      +\operatorname{tr}(\widehat{\mathbf P}_{i+1}\mathbf D_j).
\end{aligned}
\label{di_unprojected_identity}
\end{equation}

Define
$
    \widehat{\mathbf p}_{i+1}
    \triangleq
    \operatorname{svec}(\widehat{\mathbf P}_{i+1})
    \in\mathbb R^{d_n},
$
and
$
    \mathbf C
    \triangleq
    \mathcal D_x^{\top}
    \in\mathbb R^{N\times d_n}.
$
Define the unprojected actor regressor
$\mathbf H_i\in\mathbb R^{N\times sn}$ and the raw response
$\mathbf y_i^{\mathrm r}\in\mathbb R^N$ by
\begin{subequations}
\label{di_unprojected_design}
\begin{align}
    [\mathbf H_i]_{j,:}
    \triangleq
    &2[
      \gamma^2
      \operatorname{vec}\!\left(
      \mathbf S_{w,j}-\mathbf L_i\mathbf X_j\right)^{\top}\notag\\
      &
      \operatorname{vec}\!\left[
      \mathbf R(\mathbf S_{u,j}+\mathbf K_i\mathbf X_j)
      \right]^{\top}
],
    \label{di_Hi}\\
    [\mathbf y_i^{\mathrm r}]_j
    &\triangleq
    \operatorname{tr}(\mathbf Q_i\mathbf X_j).
    \label{di_raw_y}
\end{align}
\end{subequations}
Then \eqref{di_unprojected_identity} is equivalent to the
actor--critic regression
\begin{equation}
    \mathbf H_i\boldsymbol\eta_{i+1}
    -\mathbf C\widehat{\mathbf p}_{i+1}
    =
    \mathbf y_i^{\mathrm r},
\label{di_unprojected_regression}
\end{equation}
where $\widehat{\mathbf p}_{i+1}$ is a nuisance parameter.
Since
$\ker(\mathbf C^{\top})=\ker(\mathcal D_x)$,
the definitions of $\mathbf N_D$, $\mathbf\Lambda_i$, and
$\mathbf y_i$ give
\begin{equation}
    \mathbf N_D^{\top}\mathbf C=\mathbf0,
    \qquad
    \mathbf\Lambda_i=\mathbf N_D^{\top}\mathbf H_i,
    \qquad
    \mathbf y_i=\mathbf N_D^{\top}\mathbf y_i^{\mathrm r}.
\label{di_projection_relations}
\end{equation}
Thus, left multiplication by $\mathbf N_D^{\top}$ eliminates the
critic term from \eqref{di_unprojected_regression} and yields the
actor-only regression \eqref{db_gain_regression}.

\begin{theorem}
\label{thm:di_lossless_elimination}
At iteration $i$, eliminating the nuisance critic parameter from
\eqref{di_unprojected_regression} by left multiplication with
$\mathbf N_D^{\top}$ is lossless with respect to the actor variable:
\begin{equation}
\begin{aligned}
    &\left\{
      \boldsymbol\eta\in\mathbb R^{sn}:\,
      \mathbf H_i\boldsymbol\eta-\mathbf C\mathbf p
      =\mathbf y_i^{\mathrm r}
      \text{ for some }\mathbf p\in\mathbb R^{d_n}
      \right\}\\
    &\qquad=
      \left\{
      \boldsymbol\eta\in\mathbb R^{sn}:\,
      \mathbf\Lambda_i\boldsymbol\eta=\mathbf y_i
      \right\}.
\end{aligned}
\label{di_actor_set_equivalence}
\end{equation}
Consequently, the policy update $\boldsymbol\eta_{i+1}$ is uniquely
determined by the collected data if and only if
\begin{equation}
    \operatorname{rank}(\mathbf\Lambda_i)=sn.
\label{di_actor_rank_condition}
\end{equation}
Under \eqref{di_actor_rank_condition}, the critic parameter
$\widehat{\mathbf p}_{i+1}$ is uniquely determined from
\eqref{di_unprojected_regression} if and only if
\begin{equation}
    \operatorname{rank}(\mathcal D_x)=d_n.
\label{di_critic_rank_condition}
\end{equation}
Moreover,
\begin{equation}
  \operatorname{rank}\!
   \begin{bmatrix}
     \mathbf C&\mathbf H_i
    \end{bmatrix}
    =
    \operatorname{rank}(\mathcal D_x)
    +\operatorname{rank}(\mathbf\Lambda_i).
\label{di_rank_decomposition}
\end{equation}
\end{theorem}

\begin{proof}
Because the columns of $\mathbf N_D$ form an orthonormal basis of
$\ker(\mathbf C^\top)$,
$
    \operatorname{im}(\mathbf C)
    =\bigl(\ker(\mathbf C^\top)\bigr)^\perp
    =\ker(\mathbf N_D^\top).
$
Let $\boldsymbol\eta\in\mathbb R^{sn}$. If
$\mathbf H_i\boldsymbol\eta-\mathbf C\mathbf p
=\mathbf y_i^{\mathrm r}$ for some
$\mathbf p\in\mathbb R^{d_n}$, then
$\mathbf H_i\boldsymbol\eta-\mathbf y_i^{\mathrm r}
\in\operatorname{im}(\mathbf C)=\ker(\mathbf N_D^\top)$.
Equation~\eqref{di_projection_relations} therefore gives
$\mathbf\Lambda_i\boldsymbol\eta=\mathbf y_i$.
Conversely, if $\mathbf\Lambda_i\boldsymbol\eta=\mathbf y_i$,
then \eqref{di_projection_relations} gives
$
    \mathbf N_D^\top
    (\mathbf H_i\boldsymbol\eta-\mathbf y_i^{\mathrm r})
    =\mathbf0.
$
Thus,
$\mathbf H_i\boldsymbol\eta-\mathbf y_i^{\mathrm r}
\in\ker(\mathbf N_D^\top)=\operatorname{im}(\mathbf C)$,
so there exists $\mathbf p\in\mathbb R^{d_n}$ satisfying
$\mathbf H_i\boldsymbol\eta-\mathbf C\mathbf p
=\mathbf y_i^{\mathrm r}$. This proves
\eqref{di_actor_set_equivalence}.
Exactness of the data ensures that
$(\boldsymbol\eta_{i+1},\widehat{\mathbf p}_{i+1})$
satisfies \eqref{di_unprojected_regression}. Thus, the common actor
solution set in \eqref{di_actor_set_equivalence} is nonempty. It is a
singleton if and only if $\ker(\mathbf\Lambda_i)=\{\mathbf0\}$,
which is equivalent to \eqref{di_actor_rank_condition}.
Under \eqref{di_actor_rank_condition}, the compatible critic
parameters are the solutions of
$
    \mathbf C\mathbf p
    =\mathbf H_i\boldsymbol\eta_{i+1}-\mathbf y_i^{\mathrm r}.
$
This system is consistent and has a unique solution if and only if
$\ker(\mathbf C)=\{\mathbf0\}$, or 
$\operatorname{rank}(\mathbf C)=d_n$. Since
$\mathbf C=\mathcal D_x^\top$, this condition is equivalent to
\eqref{di_critic_rank_condition}.
Finally, let
$r_D\triangleq\operatorname{rank}(\mathbf C)$ and let
$\mathbf Q_D\in\mathbb R^{N\times r_D}$ have orthonormal columns
spanning $\operatorname{im}(\mathbf C)$. Since
$\operatorname{im}(\mathbf C)
=(\ker(\mathbf C^\top))^\perp$,
the matrix $[\mathbf Q_D\ \mathbf N_D]$ is orthogonal. Therefore,
\begin{equation*}
\begin{aligned}
    \operatorname{rank}
    \begin{bmatrix}
        \mathbf C&\mathbf H_i
    \end{bmatrix}
    =
    \operatorname{rank}
    \begin{bmatrix}
        \mathbf Q_D^\top\mathbf C&
        \mathbf Q_D^\top\mathbf H_i\\
        \mathbf0&
        \mathbf\Lambda_i
    \end{bmatrix}.
\end{aligned}
\end{equation*}
Because $\mathbf Q_D^\top\mathbf C$ has full row rank $r_D$, no
nonzero linear combination of the upper rows can vanish in the first
$d_n$ columns. Its row space therefore intersects that of
$\begin{bmatrix}\mathbf0&\mathbf\Lambda_i\end{bmatrix}$ only at the
origin. Thus,
$
    \operatorname{rank}
    \begin{bmatrix}
        \mathbf C&\mathbf H_i
    \end{bmatrix}
    =
    r_D+\operatorname{rank}(\mathbf\Lambda_i).
$
Since
$r_D=\operatorname{rank}(\mathbf C)
=\operatorname{rank}(\mathcal D_x)$,
\eqref{di_rank_decomposition} follows.
\end{proof}

\begin{remark}
The proof of Theorem~3 relies only on the relation
$\operatorname{im}(\mathbf C)=\ker(\mathbf N_D^{\top})$. Hence, the
endpoint projection can be interpreted as eliminating the nuisance critic
parameter by passing to the quotient space modulo
$\operatorname{im}(\mathbf C)$. More generally, any left annihilator
having kernel $\operatorname{im}(\mathbf C)$ produces the same feasible
actor set. The orthonormal choice $\mathbf N_D$ is particularly convenient
numerically. This interpretation explains why critic elimination is
lossless with respect to the policy variables.
\end{remark}

\subsection{Identifiability and Data Requirements}

Theorem~\ref{thm:di_lossless_elimination} separates the information
required for critic identification from that required for policy
identification. In particular, joint actor--critic identifiability,
\begin{equation}
    \operatorname{rank}\!
    \begin{bmatrix}
        \mathbf C&\mathbf H_i
    \end{bmatrix}
    =
    d_n+sn,
\label{di_joint_rank_condition}
\end{equation}
is equivalent to
$   \operatorname{rank}(\mathcal D_x)=d_n$ and
$
    \operatorname{rank}(\mathbf\Lambda_i)=sn.
$
The first condition ensures critic uniqueness once the policy update
is fixed, whereas the second is necessary and sufficient for unique
policy identification. Thus, the data may identify the policy update
even when they do not identify the critic.

Because $\mathbf\Lambda_i$ has
$N-\operatorname{rank}(\mathcal D_x)$ rows, a necessary dimensional
condition for \eqref{di_actor_rank_condition} is
\begin{equation}
    N
    \geq
    \operatorname{rank}(\mathcal D_x)+sn.
\label{di_sample_requirement}
\end{equation}
Condition \eqref{di_sample_requirement} is only necessary. Full
column rank of $\mathbf\Lambda_i$ also requires adequate excitation.
The stronger bound $N\geq d_n+sn$ is conservative whenever
$\operatorname{rank}(\mathcal D_x)<d_n$ and still does not guarantee
actor identifiability.
\begin{table}[t!]
\centering
\caption{Data Informativity and Identifiability Conditions}
\label{tab:informativity}
\footnotesize
\begin{tabular}{@{}p{0.47\columnwidth}p{0.45\columnwidth}@{}}
\hline
\textbf{Condition}
& \textbf{Uniquely Identified Quantity} \\
\hline

$\operatorname{rank}(\mathbf\Lambda_i)=sn$
&
Policy gains $(\mathbf L_{i+1},\mathbf K_{i+1})$
\\[1mm]

$\operatorname{rank}(\mathbf Z_{\mathrm a})=sn$
&
Policy gains $(\mathbf L_{i+1},\mathbf K_{i+1})$ at all stabilizing iterates
\\[1mm]

$\operatorname{rank}(\mathcal D_x)=d_n$
&
Value matrix $\widehat{\mathbf P}_{i+1}$, given a unique policy update
\\[1mm]

$\operatorname{rank}[\,\mathbf C\;\;\mathbf H_i\,]=d_n+sn$
&
Both $(\mathbf L_{i+1},\mathbf K_{i+1})$ and
$\widehat{\mathbf P}_{i+1}$
\\

\hline
\end{tabular}
\end{table}
\subsection{Iteration-Invariant Actor Informativity}

The next result shows that actor informativity depends only on the
fixed data batch, not on the stabilizing policy iterate.

\begin{corollary}
\label{cor:di_rank_equivalence}
For every $\mathbf\Theta_i\in\mathcal S_g$,
\begin{equation}
    \operatorname{rank}(\mathbf\Lambda_i)
    =
    \operatorname{rank}(\mathbf\Gamma_i)
    =
    \operatorname{rank}(\mathbf Z_{\mathrm a}).
\label{di_rank_equivalence}
\end{equation}
Consequently, Assumption~\ref{ass:data_rank} is equivalent to the
necessary and sufficient actor-informativity condition
\eqref{di_actor_rank_condition} at every stabilizing iterate.
\end{corollary}

\begin{proof}
Differentiating \eqref{dd_residual_relation} and using
\eqref{mb_J_definition} yields
\begin{equation}
    \mathbf\Gamma_i
    =
    2\mathbf Z_{\mathrm a}
    \mathbf J_{\mathrm a}(\mathbf\Theta_i).
\label{di_jacobian_factorization}
\end{equation}
Lemma~\ref{thm:pgre_jacobian_nonsingular} implies that
$\mathbf J_{\mathrm a}(\mathbf\Theta_i)$ is nonsingular. Therefore,
$\operatorname{rank}(\mathbf\Gamma_i)
=\operatorname{rank}(\mathbf Z_{\mathrm a})$. The coordinates
$\operatorname{vec}(\mathbf\Theta_{i+1})$ and
$\boldsymbol\eta_{i+1}$ differ by a fixed signed permutation, so
$\operatorname{rank}(\mathbf\Lambda_i)
=\operatorname{rank}(\mathbf\Gamma_i)$. This proves
\eqref{di_rank_equivalence} and the stated equivalence.
\end{proof}

Thus, actor informativity can be verified once from the fixed data
batch and remains valid throughout all stabilizing iterations,
without requiring critic identifiability. Table~\ref{tab:informativity} summarizes the resulting rank
conditions and the quantities they uniquely determine.

\section{Illustrative Example}
\label{sec:illustrative_example}
We use a power system frequency-regulation example to examine model-based equivalence, data-driven policy recovery, actor–critic identifiability, and computational scalability.

Consider the following power systems
frequency-regulation model:
\begin{subequations}
\label{ex_power_dynamics}
\begin{align}
\Delta\dot{f}(t)
={}&-\frac{1}{T_{\mathrm{p}}}\Delta f(t)
+\frac{k_{\mathrm{p}}}{T_{\mathrm{p}}}\Delta p_{\mathrm{m}}(t)
+\frac{k_{\mathrm{p}}}{T_{\mathrm{p}}}w(t),
\label{ex_frequency_dynamics}\\
\Delta\dot{p}_{\mathrm{m}}(t)
={}&-\frac{1}{T_{\mathrm{t}}}\Delta p_{\mathrm{m}}(t)
+\frac{1}{T_{\mathrm{t}}}\Delta p_{\mathrm{v}}(t),
\label{ex_turbine_dynamics}\\
\Delta\dot{p}_{\mathrm{v}}(t)
={}&-\frac{1}{\rho T_{\mathrm{g}}}\Delta f(t)
-\frac{1}{T_{\mathrm{g}}}\Delta p_{\mathrm{v}}(t)
+\frac{1}{T_{\mathrm{g}}}u(t),
\label{ex_governor_dynamics}
\end{align}
\end{subequations}
where $\Delta f(t)$ denotes the frequency deviation,
$\Delta p_{\mathrm{m}}(t)$ is the turbine mechanical-power deviation,
and $\Delta p_{\mathrm{v}}(t)$ represents the governor-state
deviation. The positive constants $k_{\mathrm{p}}$ and
$T_{\mathrm{p}}$ are the power-system gain and time constant,
respectively, while $T_{\mathrm{t}}$, $T_{\mathrm{g}}$, and $\rho$
denote the turbine time constant, governor time constant, and droop
coefficient.
Define the state vector as
$
x(t)
\triangleq
\operatorname{col}\!\left(
\Delta f(t),\Delta p_{\mathrm{m}}(t),
\Delta p_{\mathrm{v}}(t)\right)
\in\mathbb{R}^{3}.
$
Then, \eqref{ex_power_dynamics} takes the form of \eqref{1}, with
$n=3$ and $m=q=1$.
The system parameters are chosen as in \cite{wang2017intelligent}.
The state- and control-weighting matrices
and the prescribed disturbance-attenuation level are selected as
$
\textbf{Q}=I_{3}$,
$\textbf{R}=1$, and
$\gamma=5$.
The initial state is specified as
$
x_{0}
=
\begin{bmatrix}
0.1&-0.2&0.2
\end{bmatrix}^{\top}.
$ 
The optimal control gains are
\begin{eqnarray*}
\mathbf{K}^{\ast} \!\!&=&\!\!\left[
\begin{array}{ccc}
\!\!-0.5174 & \!\!0.3363 \!\!&0.4237
\end{array}
\right] \text{,}\\
\mathbf{L}^{\ast}&=&\left[
\begin{array}{ccc}
0.0282 & 0.0029 & -0.0010
\end{array}%
\right].
\end{eqnarray*}
\textbf{Validation studies}: 
We first evaluate the model-based critic-free PI method in
Algorithm~1. The anchor and initial joint policies are
chosen as
$
\mathbf{\Theta}_{\mathrm{a}}=\mathbf{\Theta}_{0}
=\operatorname{col}(\textbf{L}_{0},-\textbf{K}_{0})=\textbf{0}. 
$
After 4 iterations, the following solution is obtained
\begin{eqnarray*}
\mathbf{\Theta}_{\mathrm{4}}&=&\left[
\begin{array}{ccc}
0.0282 & 0.0029 & -0.0010\\
0.5174 & -0.3363 & -0.4237
\end{array}%
\right].
\end{eqnarray*}%
To verify the step-by-step equivalence between Algorithm 1 and SPUA, both methods are initialized with the same stabilizing policy $\mathbf{\Theta}_{0}$. As shown in Fig. 1, the corresponding entries of $\textbf{K}_{i}$ and $\textbf{L}_{i}$ generated by the two algorithms exhibit identical trajectories at every iteration and converge to the same optimal gains. 

Under the same initialization condition, we then evaluate the data-driven
critic-free PI method in Algorithm~2.
Fig.~2 illustrates the endpoint-projection mechanism.
Fig.~2(a) shows representative intervals
$\mathcal I_j=[t_{j-1},t_j]$, whose endpoint states are used to construct
the endpoint increments $\mathbf D_j$.
Fig.~2(b) shows the coefficients $\alpha_{j1}$ of a representative
nullspace direction satisfying
$
\sum_{j=1}^{N}\alpha_{j1}\mathbf D_j=\mathbf 0
$.
Fig.~2(c) displays the corresponding partial sums of the weighted
endpoint increments, which vanish after the complete combination is formed.
Fig.~2(d) numerically verifies this cancellation using
$\mathbf P$ only for illustration.

\begin{figure}[h]
\centering
\subfigure[Algorithm 1]{
		\includegraphics[width=4cm]{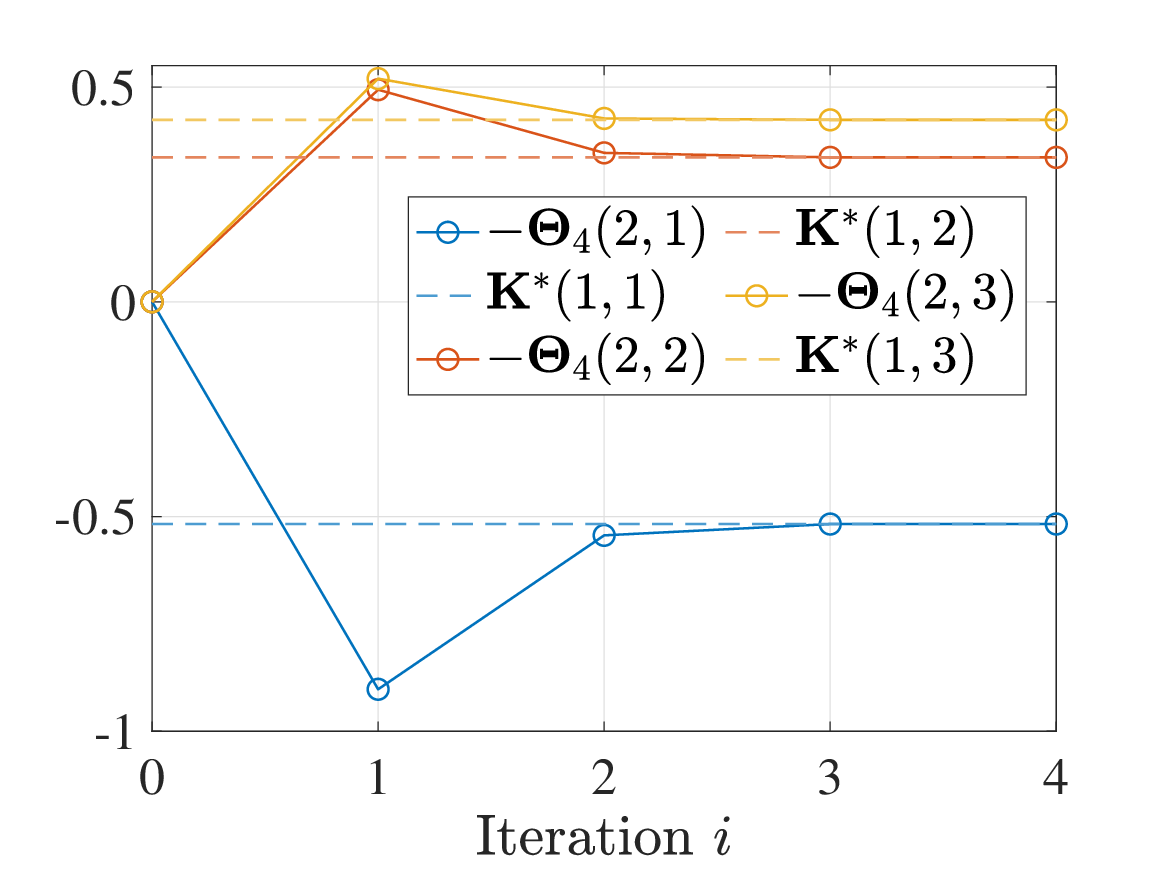}
			} 
\subfigure[Algorithm 1]{
		\includegraphics[width=4cm]{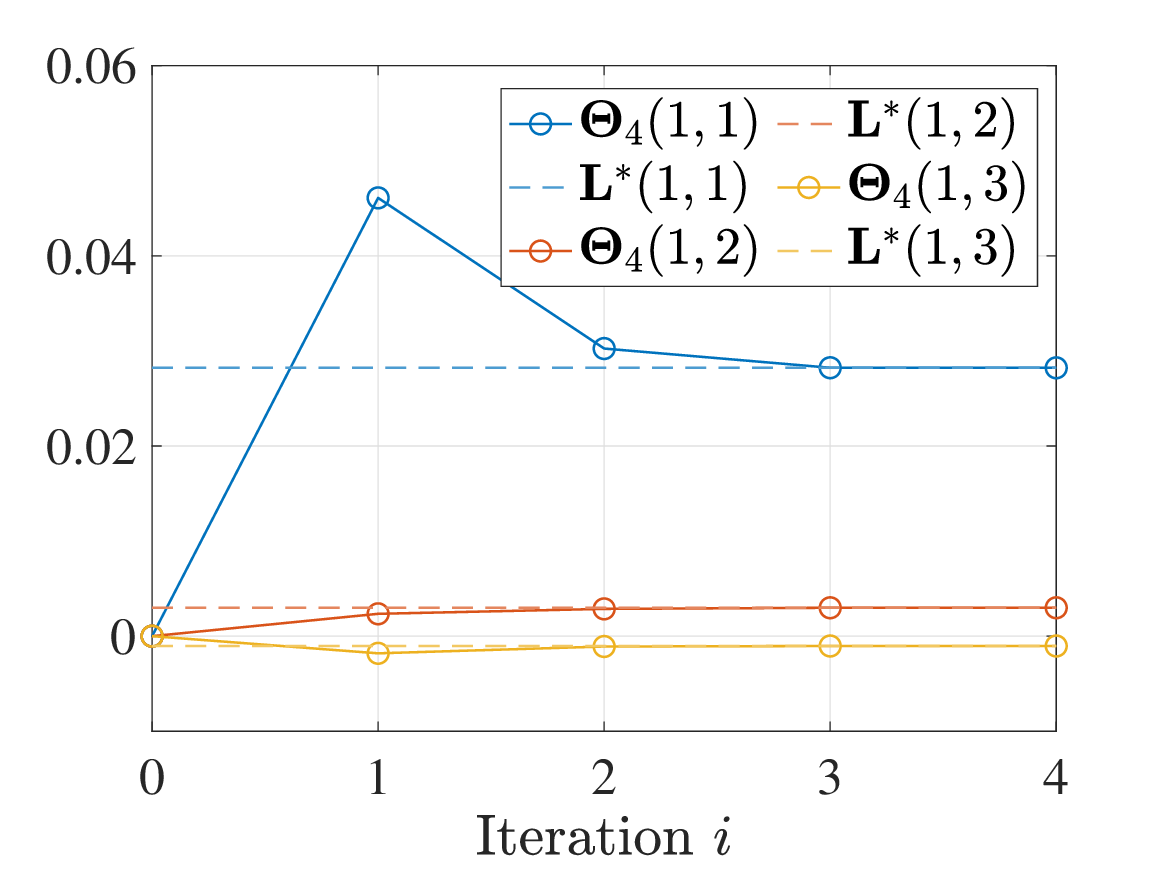}
	} 
\subfigure[SPUA]{
		\includegraphics[width=4cm]{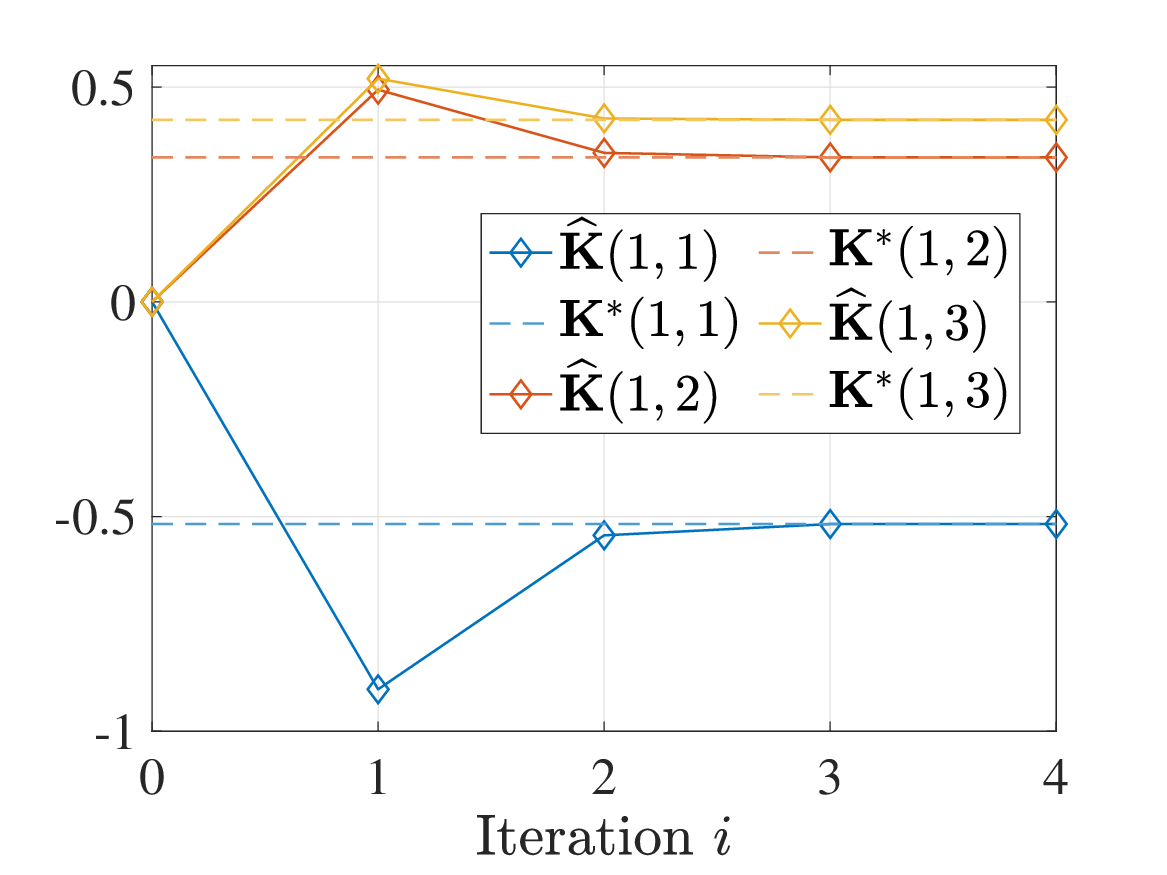}
	} 
\subfigure[SPUA]{
		\includegraphics[width=4cm]{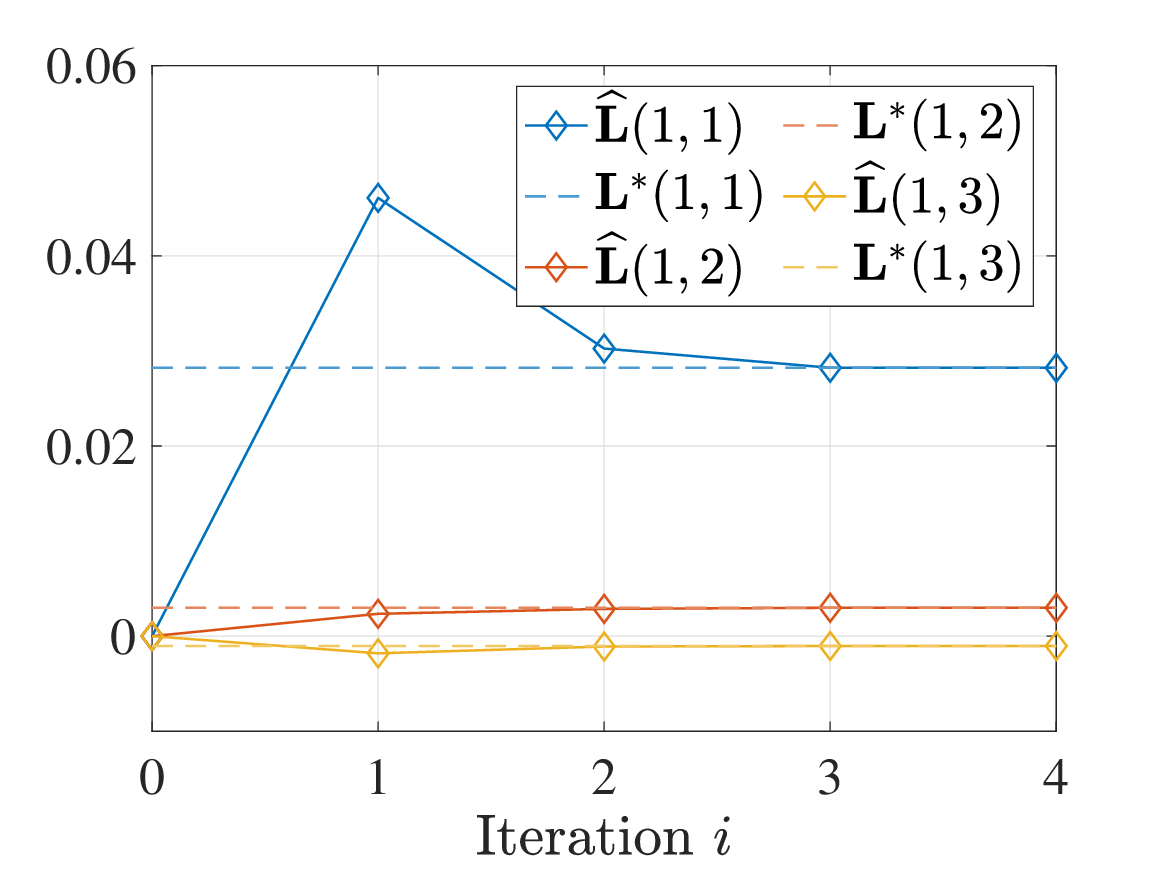}
	}
\caption{Convergence comparison between Algorithm 1 and SPUA: (a) control gain $\textbf{K}_{i}$ generated by Algorithm 1; (b) disturbance gain $\textbf{L}_{i}$ generated by Algorithm 1; (c) control gain $\widehat{\textbf{K}}_{i}$ generated by SPUA; and (d) disturbance gain $\widehat{\textbf{L}}_{i}$ generated by SPUA. The dashed lines denote the optimal gains $\mathbf{K}^{\ast}$ and $\mathbf{L}^{\ast}$.}
\label{fig1}
\end{figure}

\begin{figure}[th!]
\centering
\subfigure[]{
		\includegraphics[width=4cm]{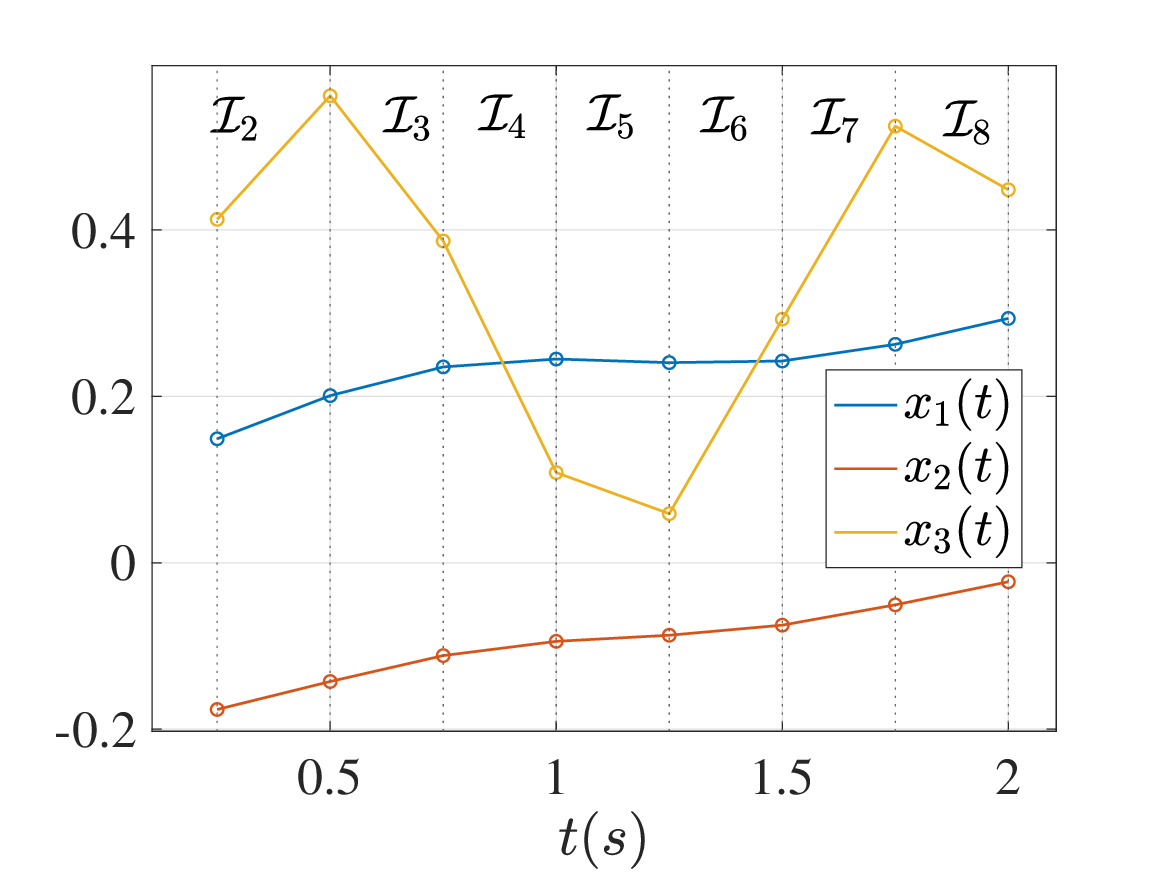}
			} 
\subfigure[]{
		\includegraphics[width=4cm]{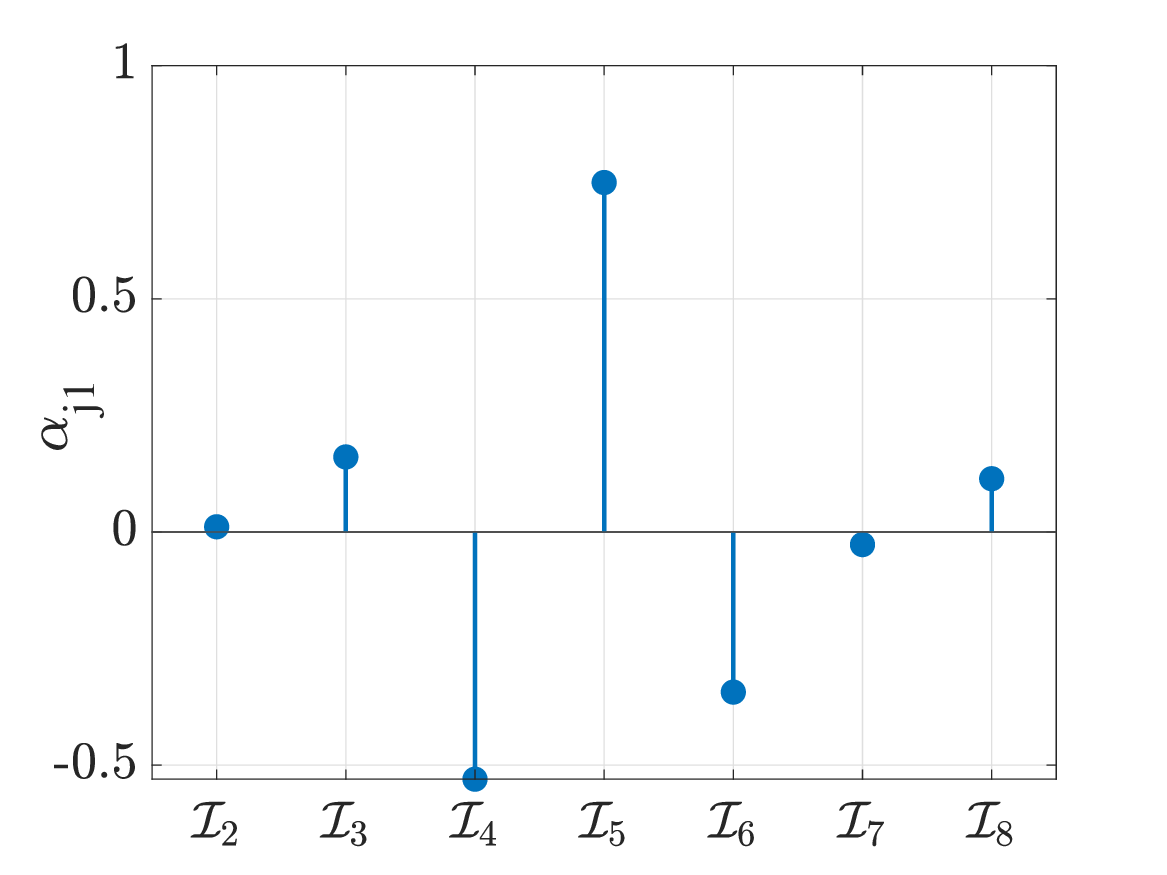}
	} 
\subfigure[]{
		\includegraphics[width=4cm]{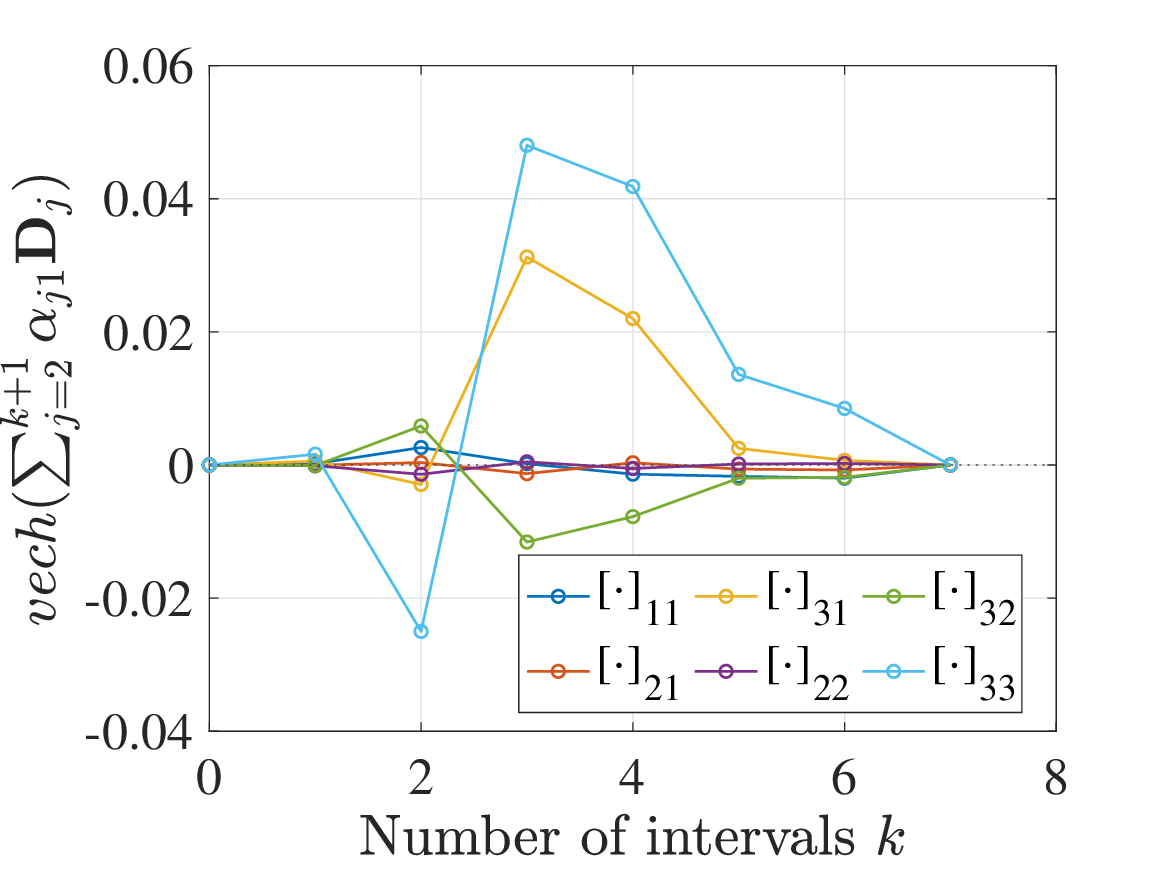}
	} 
\subfigure[]{
		\includegraphics[width=4cm]{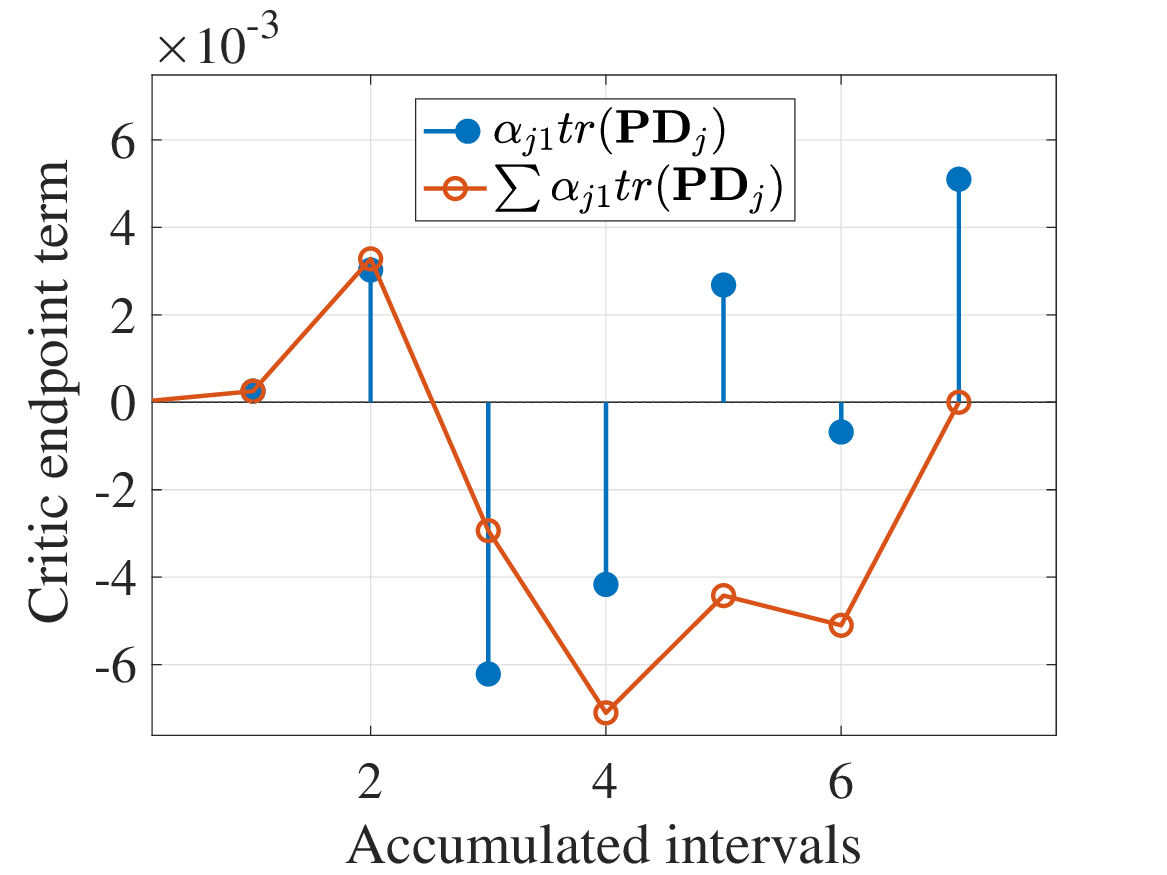}
	}
\caption{Endpoint projection for critic elimination in Algorithm 2: (a) state trajectories over representative sampling intervals $\mathcal I_j$ used to construct $\mathbf D_j$; (b) coefficients $\alpha_{j1}$ of a nullspace direction satisfying $
\sum_{j=1}^{N}\alpha_{j1}\mathbf D_j=\mathbf 0
$; (c) partial sums of the weighted endpoint increments $
\sum_{j=1}^{N}\alpha_{j1}\mathbf D_j
$; and (d) corresponding critic endpoint terms and their accumulated sum, which vanishes after all intervals are combined.}
\label{fig1}
\end{figure}

\begin{figure}[t]
\centering
\includegraphics[width=0.9\linewidth]{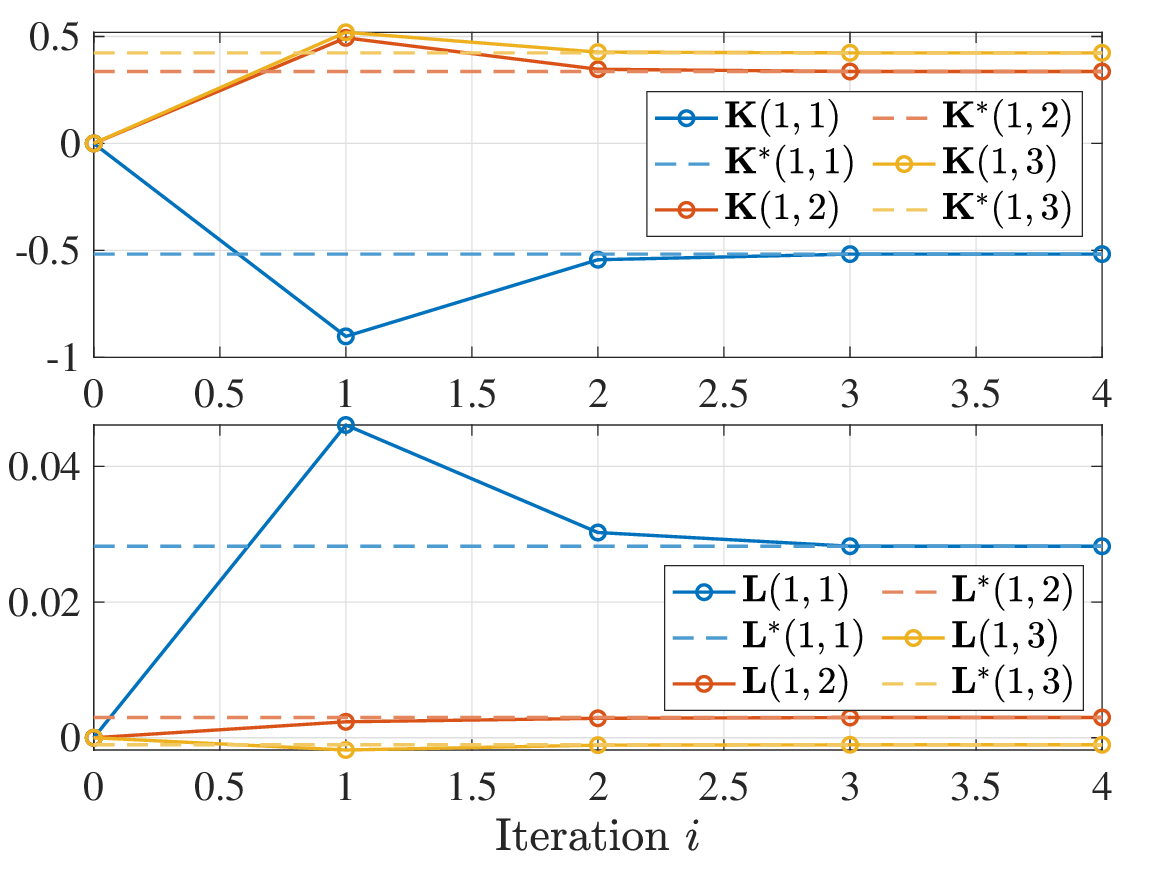}
\caption{Convergence of control gains $\textbf{K}_{i}$ and $\textbf{L}_{i}$ generated by Algorithm 2}
\end{figure}
\begin{figure}[t]
\centering
\includegraphics[width=0.92\linewidth]{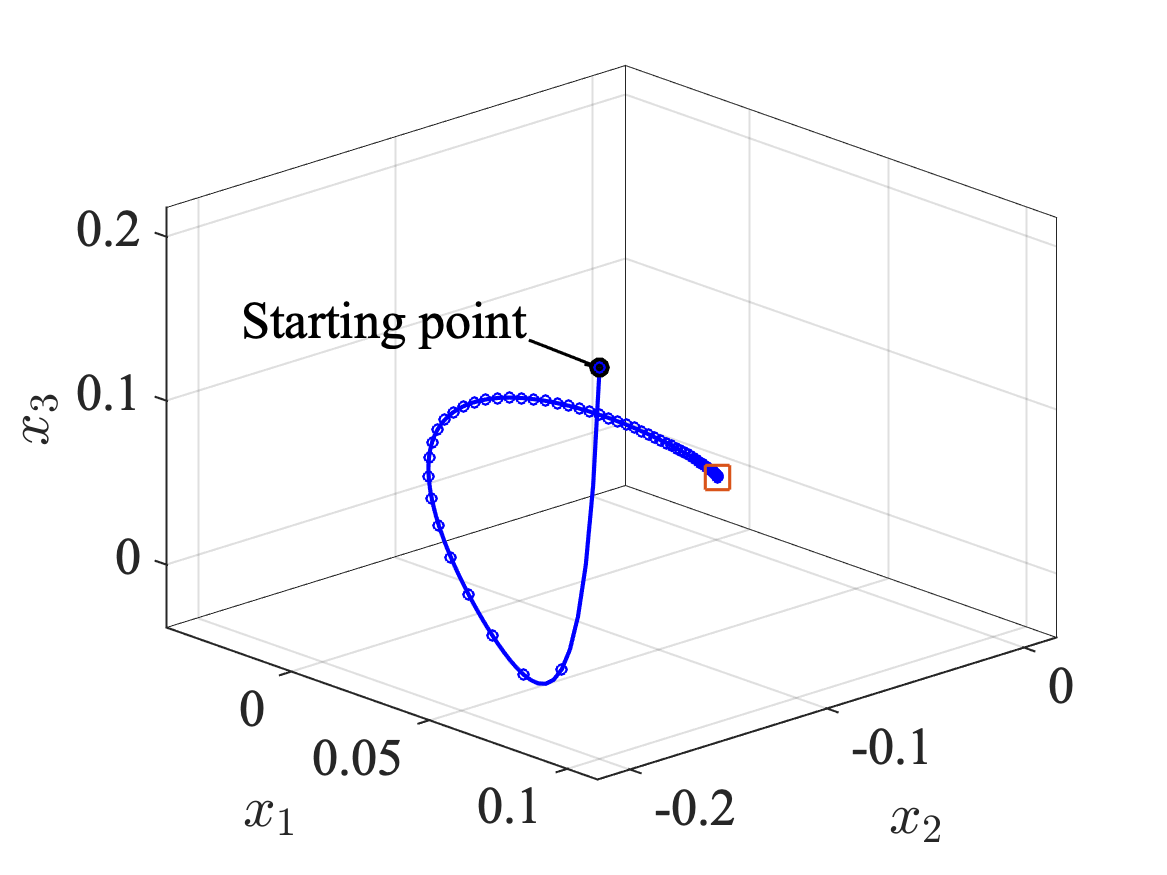}
\caption{System state trajectory under the control gains $\textbf{K}_{4}$ and $\textbf{L}_{4}$ learned by Algorithm 2. The trajectory starts from $
x_{0}
=
\begin{bmatrix}
0.1&-0.2&0.2
\end{bmatrix}^{\top},
$ and converges to the origin.}
\end{figure}

Fig. 3 shows the convergence of the policy gains generated by Algorithm 2. The control gains $\textbf{K}_{i}$ and $\textbf{L}_{i}$ converge rapidly to $\mathbf{K}^{\ast}$ and $\mathbf{L}^{\ast}$, respectively, with convergence obtained after four iterations. This result verifies that the proposed data-driven critic-free update successfully recovers the optimal saddle-point policy without explicitly estimating the value matrix. Fig. 4 depicts the state trajectory corresponding to $\textbf{K}_{4}$ and $\textbf{L}_{4}$. Starting from the prescribed initial condition, all state components converge toward the equilibrium.

\textbf{Comparative studies:}
To illustrate the distinct roles of actor informativity and critic
identifiability characterized in
Section~\ref{sec:data_informativity}, we construct three data batches with
different rank profiles while keeping the plant, game parameters, initial
condition, and anchor policy unchanged. Here,
$s=m+q=2$, $sn=6$, and $d_n=n(n+1)/2=6$.

\noindent \textbf{Case I} \emph{(actor- and critic-informative):}
A persistently excited off-policy batch is collected over
$t\in[0,8]~\mathrm{s}$ with sampling interval $T=0.25~\mathrm{s}$,
giving $N=32$ data intervals. The control and disturbance channels are
excited by
\begin{equation*}
\begin{aligned}
\mathbf u_{\mathrm b}(t)
={}&0.8\sin(0.7t)+0.6\sin(1.9t)\\
   &+0.4\cos(3.1t)+0.25\sin(4.7t),\\
\mathbf w_{\mathrm b}(t)
={}&0.7\cos(0.5t)+0.5\sin(1.3t)\\
   &+0.35\cos(2.7t)+0.25\sin(4.1t).
\end{aligned}
\end{equation*}
The resulting batch satisfies
$\operatorname{rank}(\mathbf Z_{\mathrm a})
=\operatorname{rank}(\mathbf\Lambda_0)
=\operatorname{rank}(\mathcal D_x)=6$, and is therefore informative for
both the policy update and the critic.

\noindent \textbf{Case II} \emph{(actor-informative but critic-nonidentifying):}
The batch consists of one transition interval from $\mathbf x_0$ to the
origin, followed by twelve plant-consistent zero-endpoint loops. For each
loop interval $\mathcal I_j=[t_{j-1},t_j]$,
$\mathbf x(t_{j-1})=\mathbf x(t_j)=\mathbf0$ and thus
$\mathbf D_j=\mathbf0$. Only the initial transition contributes to the
endpoint matrix, giving
$\operatorname{rank}(\mathcal D_x)=1<d_n$. The intra-interval state and
input variations remain sufficiently rich to yield
$\operatorname{rank}(\mathbf Z_{\mathrm a})
=\operatorname{rank}(\mathbf\Lambda_0)=6$. Hence, the policy update is
uniquely identifiable, whereas the critic is not.

\noindent \textbf{Case III} \emph{(actor-uninformative):}
Using the same horizon and sampling interval as in Case~I, the probing
signals are removed by setting
$\mathbf u_{\mathrm b}(t)=\mathbf w_{\mathrm b}(t)=\mathbf0$. Since
$\mathbf\Theta_{\mathrm a}=\mathbf0$, this gives
$\mathbf v_{\mathrm a}(t)=\mathbf0$ and
$\mathbf Z_{\mathrm a}=\mathbf0$. Consequently,
$\operatorname{rank}(\mathbf Z_{\mathrm a})
=\operatorname{rank}(\mathbf\Lambda_0)=0$, while the autonomous state
evolution gives $\operatorname{rank}(\mathcal D_x)=6$. Thus, full-rank
endpoint information does not make the policy update unique.

\begin{figure}[t]
\centering
\includegraphics[width=\columnwidth]{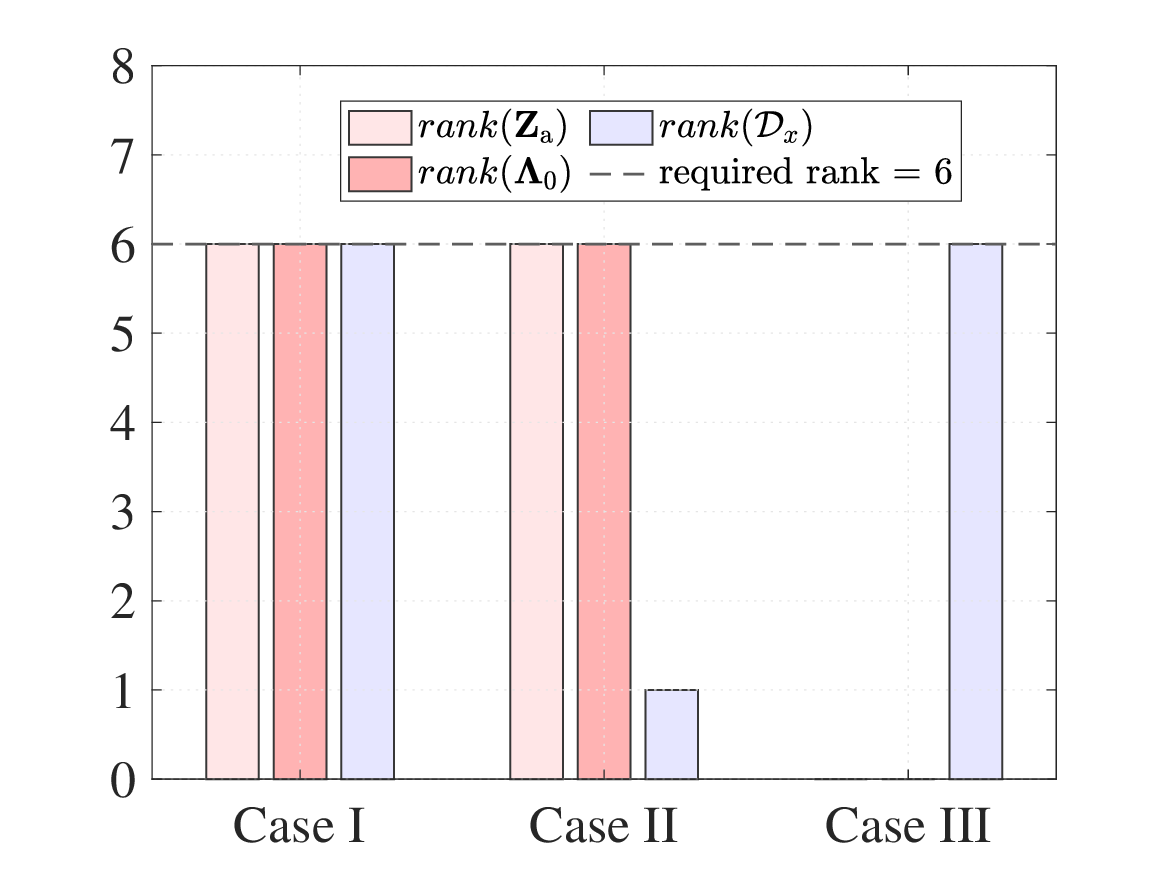}
\caption{Numerical ranks under the three data-informativity regimes. The
dashed line denotes the full-rank threshold $sn=d_n=6$.}
\label{fig:three_case_rank}
\end{figure}
\begin{table}[t]
\centering
\caption{Identifiability and policy-recovery errors for the three data regimes}
\label{tab:three_case_results}
\scriptsize
\setlength{\tabcolsep}{3pt}
\renewcommand{\arraystretch}{1.15}

\begin{tabular}{@{}ccccc@{}}
\toprule
& \multicolumn{2}{c}{Identifiability}
& \multicolumn{2}{c}{Gain error}\\
\cmidrule(lr){2-3}
\cmidrule(l){4-5}
Case
& Policy
& Critic
& $\|\mathbf K-\mathbf K^{\ast}\|$
& $\|\mathbf L-\mathbf L^{\ast}\|$\\
\midrule
I
& Identifiable
& Identifiable
& $4.0572\times10^{-10}$
& $5.3504\times10^{-11}$\\
II
& Identifiable
& Nonidentifiable
& $4.0411\times10^{-12}$
& $2.1995\times10^{-13}$\\
III
& Nonunique
& --
& --
& --\\
\bottomrule
\end{tabular}
\end{table}
Fig.~\ref{fig:three_case_rank} summarizes the three rank profiles, and
Table~\ref{tab:three_case_results} reports the associated identifiability
properties and policy-recovery errors. Cases~I and~II both satisfy the
actor rank condition. Algorithm~2 therefore yields a unique policy update
and recovers $(\mathbf K^{\ast},\mathbf L^{\ast})$ to numerical precision
in both cases, although the critic is nonidentifiable in Case~II. This
comparison shows that critic identifiability is not required for unique
policy recovery.

Case~III provides the complementary result. Its endpoint matrix has full
rank, but its actor regression is rank deficient and the policy update is
nonunique. Cases~II and~III together show that policy identifiability is
governed by actor informativity rather than critic identifiability. These
results corroborate the analysis in
Section~\ref{sec:data_informativity} and support the data-driven
critic-free formulation in Section~\ref{sec:data_based_critic_free}.

\textbf{Scalability studies:}
To assess the computational benefit of critic elimination, we consider a
family of systems with increasing state dimensions
$
n\in\{5,10,20,30,40,50\},
$
with $m=q=1$ and hence $s=m+q=2$.
The system matrix $\mathbf A$ is chosen from a stable diagonal family
whose eigenvalues are uniformly distributed over $[-2,-0.5]$.
The input and disturbance matrices $\mathbf B_1$ and $\mathbf B_2$ are
generated as normalized dense vectors. We set
$\mathbf Q=\mathbf I_n$, $\mathbf R=1$, and $\gamma=5$, and use
$\mathbf\Theta_{\mathrm a}=\mathbf 0$ as the stabilizing anchor policy.

For each $n$, the critic-free and joint actor--critic formulations use the
same data batch. The number of data intervals is chosen as
$
N=d_n+2sn,$
$
d_n=\frac{n(n+1)}{2}.
$
The generated batches satisfy
$
\operatorname{rank}(\mathcal D_x)=d_n,$
$
\operatorname{rank}(\mathbf\Lambda_i)=sn,
$
so that both formulations admit uniquely identifiable parameter vectors.
Moreover, the dimension of the nullspace of $\mathcal D_x$ is
\[
\nu
=N-\operatorname{rank}(\mathcal D_x)
=N-d_n
=2sn.
\]
The data-generation time is excluded because it is common to both
formulations. The endpoint-nullspace projection required by the
critic-free formulation is counted as a one-time preprocessing cost, and
the policy-update regressions in both formulations are solved using the
same least-squares solver.

\begin{table}[t]
\centering
\caption{Numbers of unknowns and regression-matrix dimensions for the
critic-free (CF) and joint actor--critic (A--C) formulations.}
\label{tab:computational_dimensions}
\footnotesize
\setlength{\tabcolsep}{12pt}
\renewcommand{\arraystretch}{1.15}
\begin{tabular}{@{}ccccc@{}}
\toprule
$n$
& \shortstack{Actor-only\\unknowns}
& \shortstack{Joint A-C\\unknowns}
& \shortstack{CF matrix\\size}
& \shortstack{Joint A-C\\matrix size} \\
\midrule
 5  & \textbf{10}  & 25   & $\textbf{20}\times \textbf{10}$   & $35\times 25$     \\
10  & \textbf{20}  & 75   & $\textbf{40}\times \textbf{20}$   & $95\times 75$     \\
20  & \textbf{40}  & 250  & $\textbf{80}\times \textbf{40}$   & $290\times 250$   \\
30  & \textbf{60}  & 525  & $\textbf{120}\times \textbf{60}$  & $585\times 525$   \\
40  & \textbf{80}  & 900  & $\textbf{160}\times \textbf{80}$  & $980\times 900$   \\
50  & \textbf{100} & 1375 & $\textbf{200}\times \textbf{100}$ & $1475\times 1375$ \\
\bottomrule
\end{tabular}
\end{table}

Table~\ref{tab:computational_dimensions} and
Fig.~\ref{fig:scalability}(a) show the dimensional advantage of the
critic-free formulation. The actor-only regression contains
$
sn=2n
$
unknowns, whereas the joint actor--critic regression contains
$
sn+d_n
=
2n+\frac{n(n+1)}{2}.
$
Thus, for fixed input and disturbance dimensions, the number of
critic-free unknowns grows linearly with $n$, while the joint
actor--critic parameter dimension grows quadratically. At $n=50$, the
two formulations involve $100$ and $1375$ unknowns, respectively.

The reduction in unknowns also leads to substantially smaller regression
matrices. Since $\nu=2sn$, the projected critic-free regression has
dimensions
$
\mathbf\Lambda_i\in\mathbb R^{2sn\times sn},
$
whereas the joint actor--critic coefficient matrix has dimensions
$
\begin{bmatrix}
\mathbf H_i & -\mathbf C
\end{bmatrix}
\in
\mathbb R^{(d_n+2sn)\times(d_n+sn)}.
$
For $n=50$, these dimensions are $200\times100$ and
$1475\times1375$, respectively. As shown in
Fig.~\ref{fig:scalability}(c), the corresponding regression working-memory
requirements are approximately $0.15$~MB and $15.5$~MB, amounting to a
reduction of about two orders of magnitude.

The same dimensional advantage is reflected in computation time.
Fig.~\ref{fig:scalability}(b) shows that, apart from minor nonmonotonic
variations at small $n$, the joint actor--critic update becomes
increasingly expensive as the state dimension grows, whereas the
critic-free update scales much more moderately. Consequently, the
per-update speedup in Fig.~\ref{fig:scalability}(d) increases from a
few-fold in low-dimensional systems to several tens at $n=50$.

\begin{figure}[t]
\centering
\subfigure[Number of unknowns]{
    \includegraphics[width=4cm]{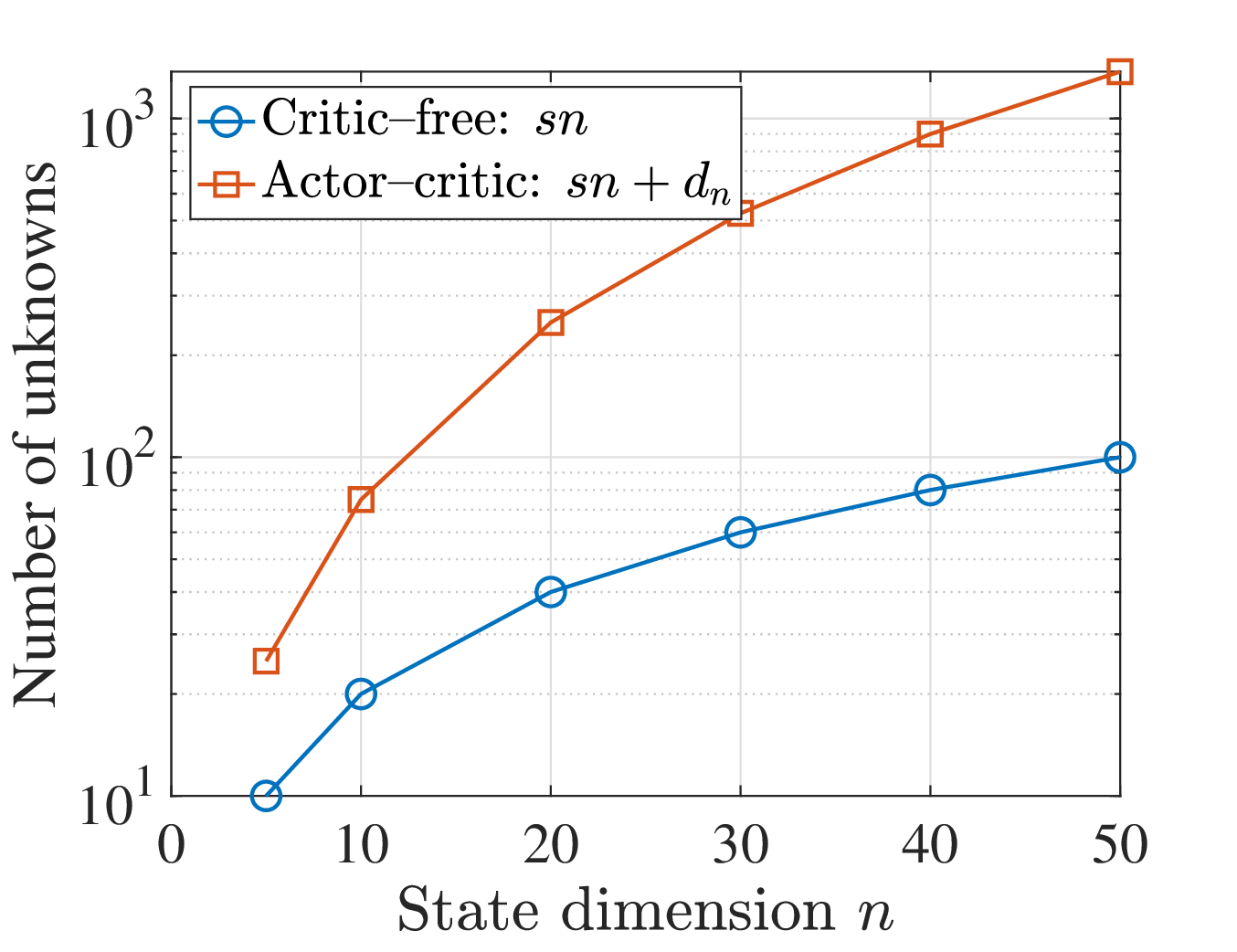}
}
\subfigure[Median CPU time]{
    \includegraphics[width=4cm]{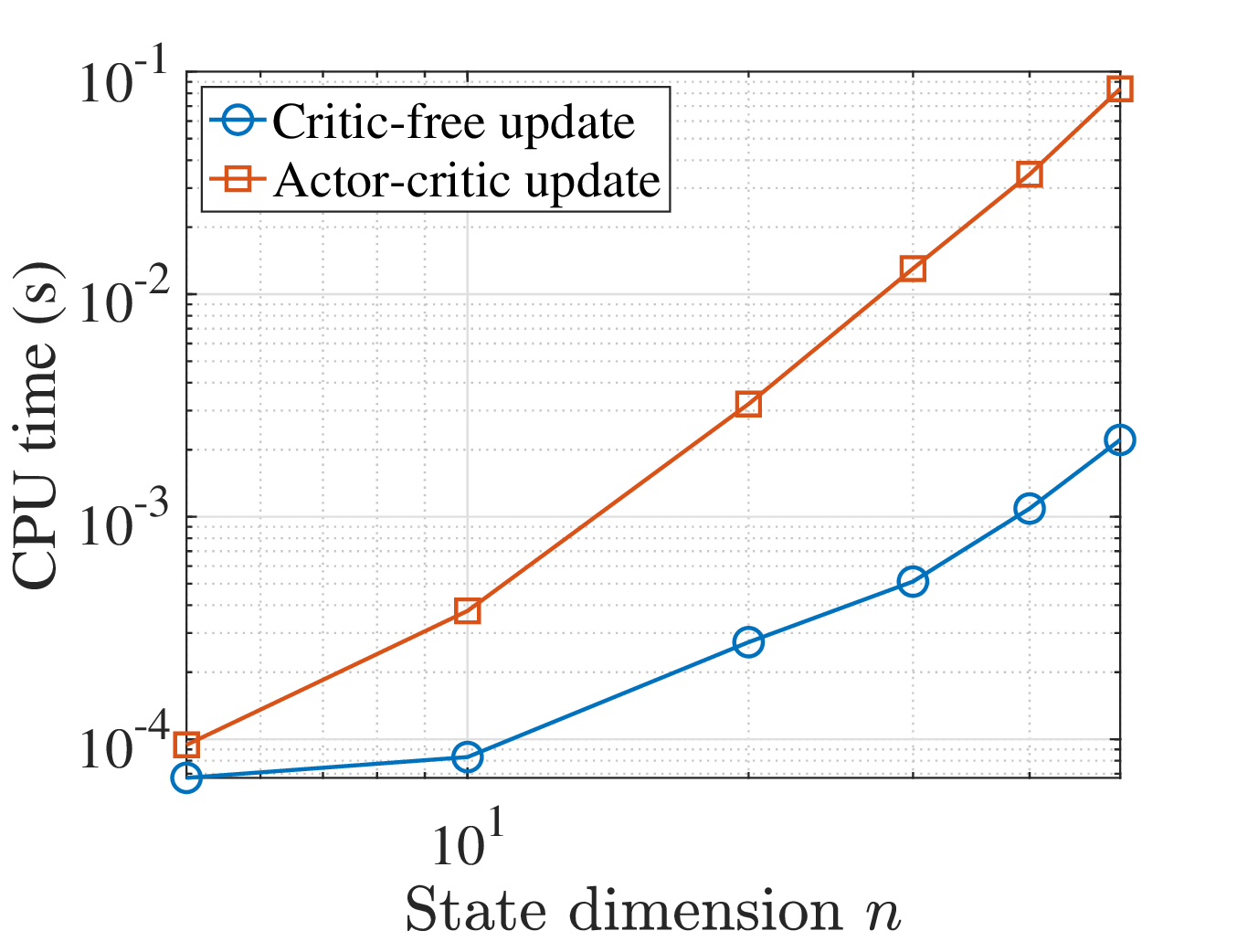}
}
\subfigure[Regression working memory]{
    \includegraphics[width=4cm]{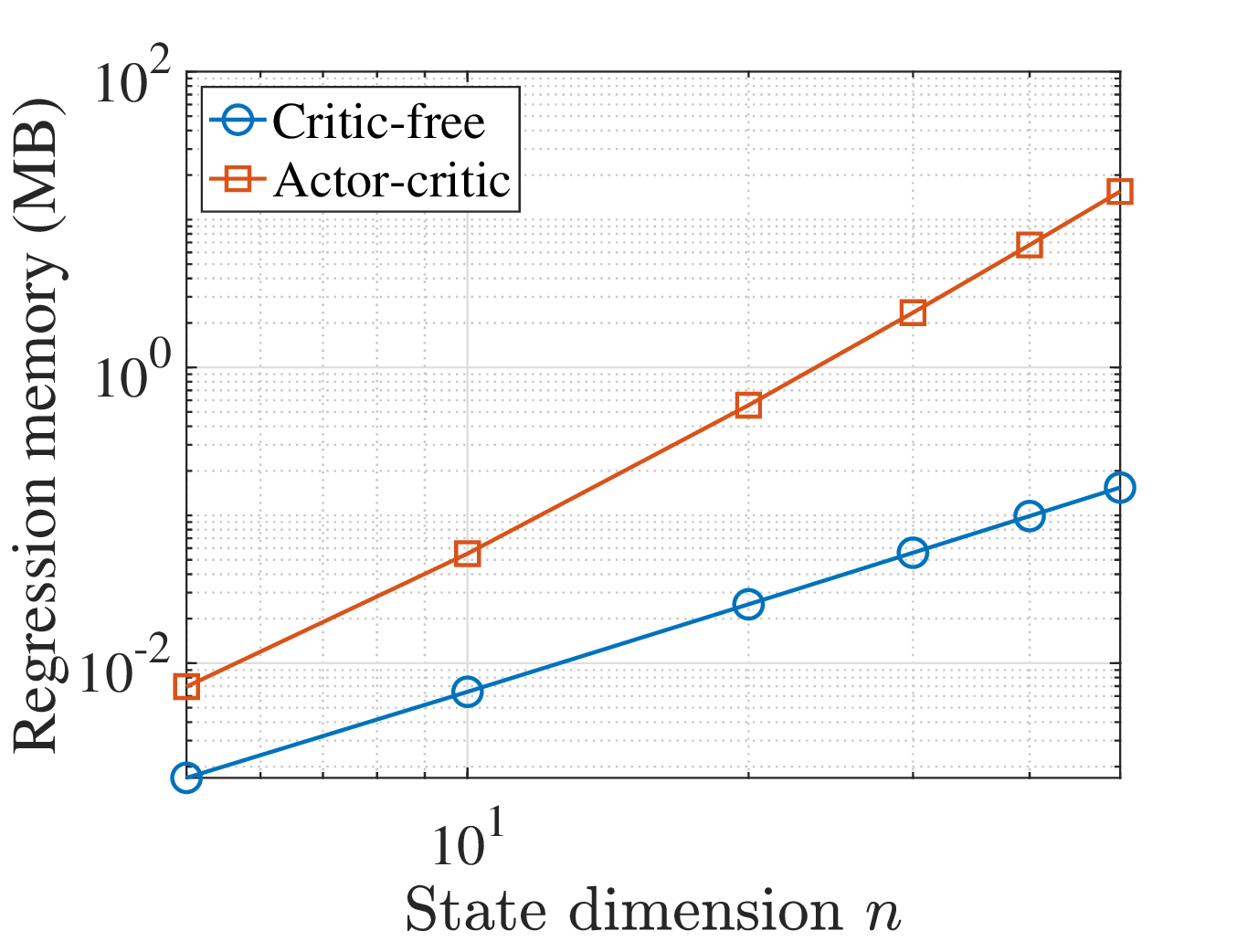}
}
\subfigure[Per-update speedup]{
    \includegraphics[width=4cm]{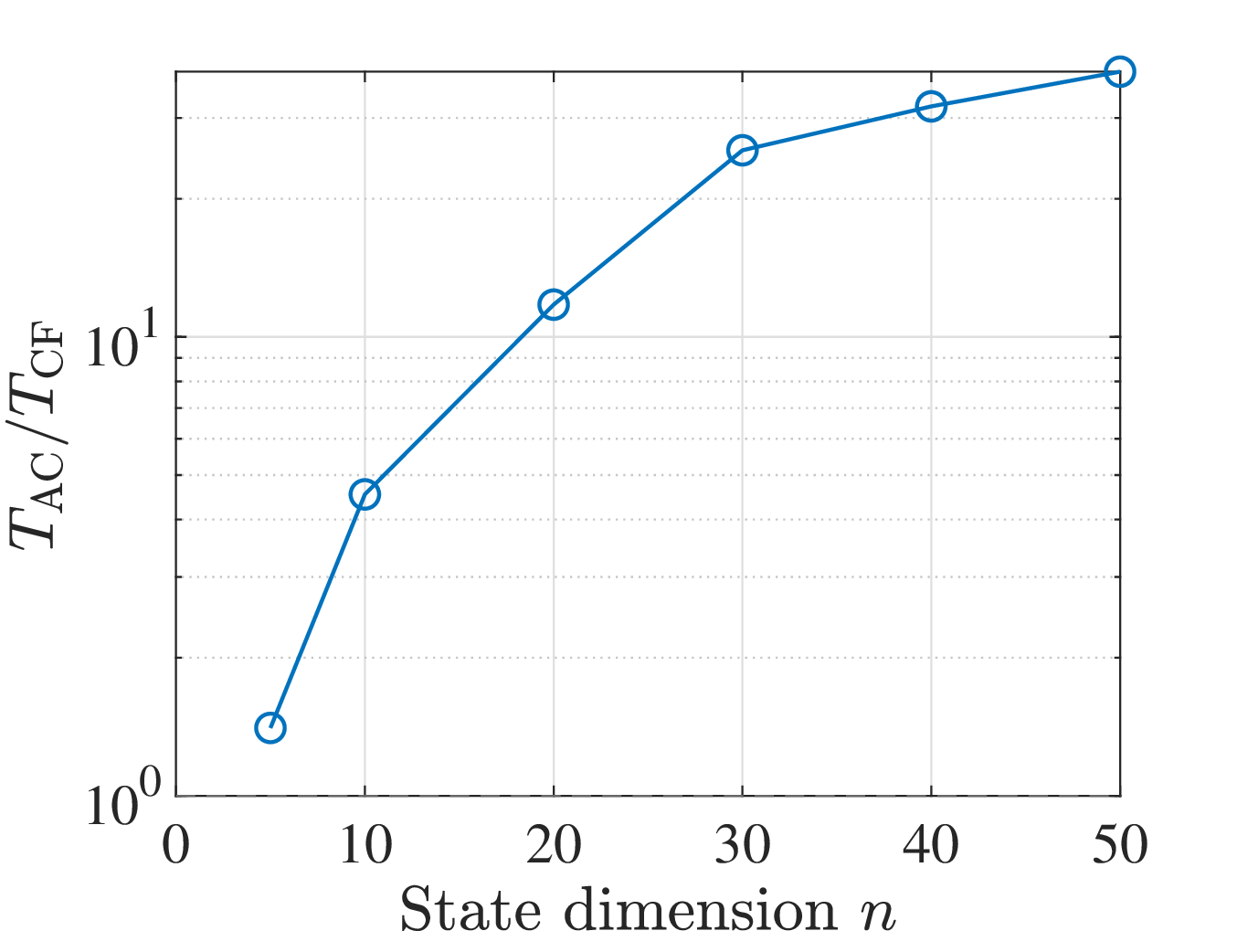}
}
\caption{Scalability comparison between the proposed critic-free and
joint actor--critic formulations: (a) number of unknown variables;
(b) median computational time of the recurring policy updates;
(c) regression working memory; and (d) per-update computational speedup.}
\label{fig:scalability}
\end{figure}

While Fig.~\ref{fig:scalability}(c) concerns the working memory of a
single regression, Fig.~\ref{fig:retained_memory} compares the sufficient
statistics that must be retained across policy iterations. After the
endpoint-nullspace projection, the critic-free formulation only retains
$
\left\{
\overline{\mathbf X}_{\ell},
\overline{\mathbf S}_{w,\ell},
\overline{\mathbf S}_{u,\ell}
\right\}_{\ell=1}^{\nu},
$
where the overbar denotes the corresponding nullspace-projected statistic.
In contrast, the joint actor--critic formulation retains the raw interval
statistics together with the endpoint data $\mathcal D_x$ required for
critic estimation. Under the scaling $N=d_n+2sn$, the storage requirements
of the projected and raw representations grow approximately as
$\mathcal O(n^3)$ and $\mathcal O(n^4)$, respectively. At $n=50$, they
require approximately $4$~MB and $44$~MB, corresponding to an
approximately $11$-fold reduction.

\begin{figure}[t]
\centering
\includegraphics[width=\columnwidth]{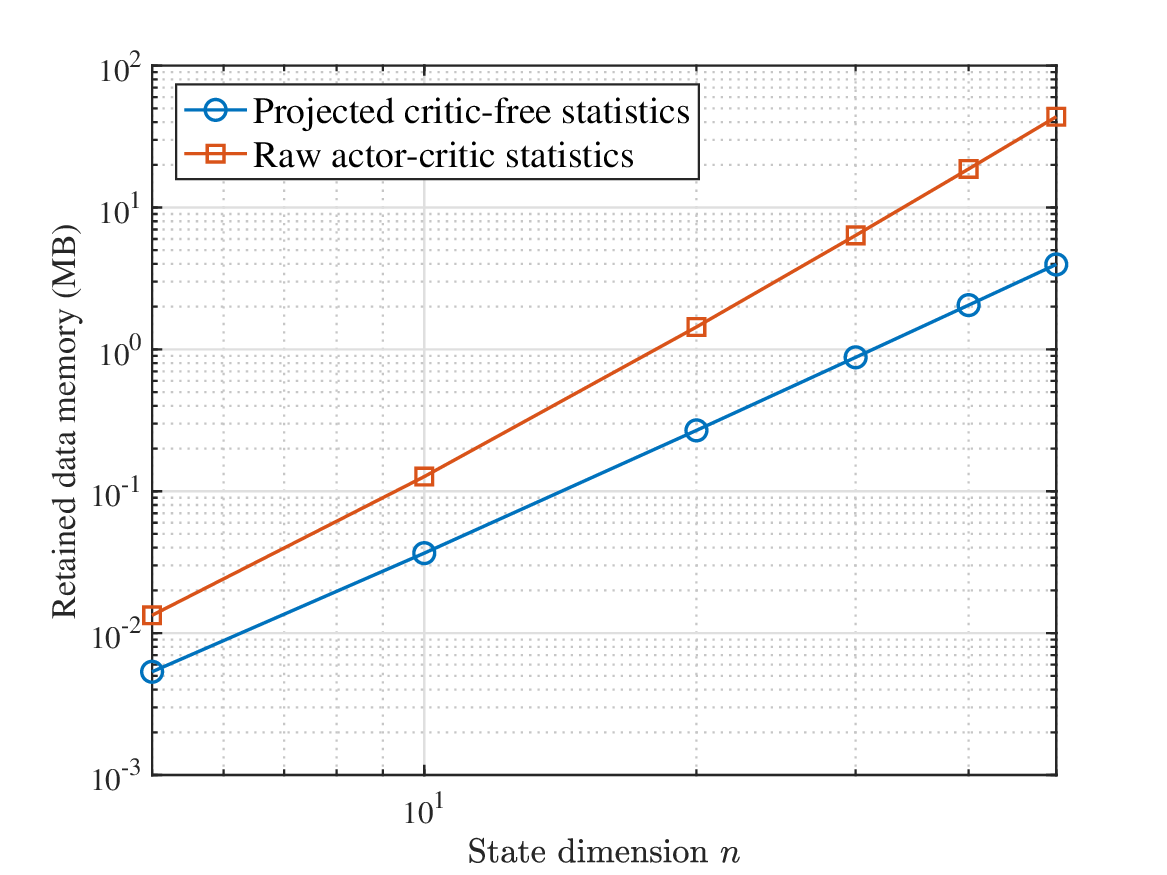}
\caption{Retained sufficient-statistics memory after preprocessing.
The proposed formulation retains only the projected actor-only statistics
required by subsequent policy updates, whereas the joint actor--critic
formulation retains the raw interval and endpoint statistics required for
joint actor and critic estimation.}
\label{fig:retained_memory}
\end{figure}
\begin{figure}[th!]
\centering
\includegraphics[width=\columnwidth]{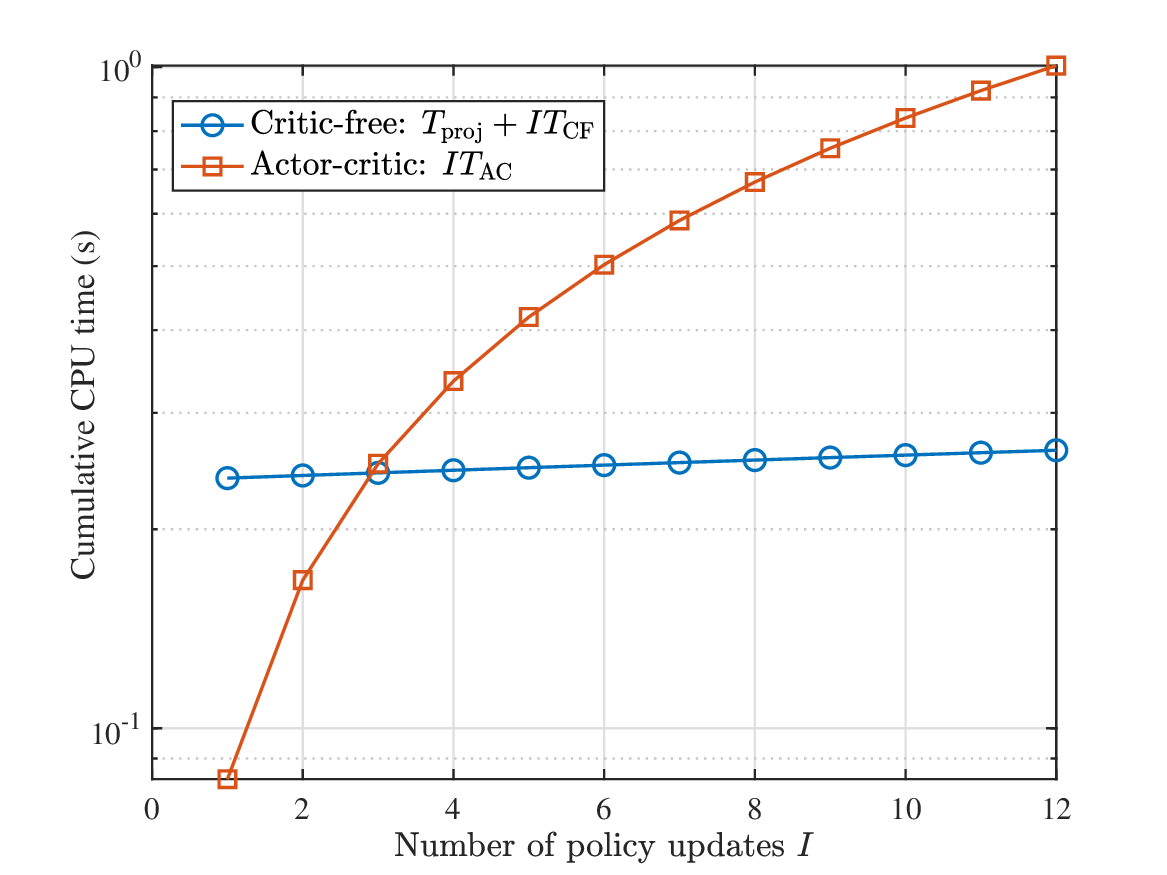}
\caption{Amortization of the one-time endpoint-projection cost at
$n=50$. The cumulative critic-free cost is estimated as
$T_{\mathrm{proj}}+I T_{\mathrm{CF}}$, whereas that of the joint
actor--critic formulation is estimated as $I T_{\mathrm{AC}}$. The
projection overhead is amortized after approximately three policy
updates.}
\label{fig:projection_amortization}
\end{figure}
The computational savings above come at the expense of the one-time
endpoint-nullspace projection. To quantify how quickly this preprocessing
cost is amortized, Fig.~\ref{fig:projection_amortization} reports the
estimated cumulative computation time at $n=50$.
$
T_{\mathrm{CF}}(I)
=T_{\mathrm{proj}}+I T_{\mathrm{CF}},$ and
$
T_{\mathrm{AC}}(I)
=I T_{\mathrm{AC}},
$
where $I$ is the number of policy updates, $T_{\mathrm{proj}}$ is the
one-time projection cost, and $T_{\mathrm{CF}}$ and $T_{\mathrm{AC}}$
are the respective per-update computation times.
Although the critic-free formulation is initially more expensive because
of the projection step, this overhead is amortized after approximately
three policy updates. Beyond this break-even point, its lower recurring
update cost yields an increasing cumulative computational advantage.
Overall, critic elimination reduces both the dimension of the recurring
policy-update regression and the memory required across iterations, with
the one-time preprocessing cost rapidly amortized as the number of policy
updates increases.

\section{Conclusion}
The proposed critic-free method shows that the value matrix is not 
intrinsically required for policy improvement in continuous-time linear 
zero-sum games. Conventional PI uses the value matrix as an intermediate 
variable to construct the next policy; under the stated conditions, the 
PGRE determines the same update directly in the joint policy space. In 
the data-driven setting, endpoint projection removes the value-matrix 
term while retaining the information required for that update. Hence, 
the usual full-rank condition for a joint actor--critic regression is 
stronger than policy improvement requires: the policy update can be 
unique even when the critic is not identifiable. Because the number of 
critic parameters grows quadratically with the state dimension, this 
separation yields increasing computational and memory benefits. 
Numerical studies corroborate the identifiability result and demonstrate 
these benefits.

\bibliographystyle{IEEEtran}
\bibliography{shwj-workwjc}
\end{document}